\documentclass[sigconf,authorversion,nonacm]{acmart}
\usepackage{fancyhdr}

\usepackage[T1]{fontenc}
\usepackage{comment}

\setcopyright{acmlicensed} %
\copyrightyear{2018} %
\acmYear{2018} %
\acmDOI{XXXXXXX.XXXXXXX} %
\acmConference[Conference acronym 'XX]{Make sure to enter the correct
  conference title from your rights confirmation email}{June 03--05,
  2018}{Woodstock, NY}  %
\acmISBN{978-1-4503-XXXX-X/2018/06}  %

\usepackage[createShortEnv]{proof-at-the-end}
\usepackage{amsmath,amsfonts}

\usepackage{amssymb}
\usepackage{amsthm}

\usepackage{tabularx}

\usepackage{textcomp}
\usepackage{xcolor}

\def\BibTeX{{\rm B\kern-.05em{\sc i\kern-.025em b}\kern-.08em
    T\kern-.1667em\lower.7ex\hbox{E}\kern-.125emX}}
\usepackage{nicefrac}
\usepackage{siunitx}
\usepackage{array,framed}
\usepackage{hyperref}
\usepackage{booktabs}
\usepackage{appendix}

\usepackage{cleveref}

\usepackage{
  color,
  float,
  epsfig,
  wrapfig,
  graphics,
  graphicx,
  subcaption,
  verbatim
}
\usepackage{setspace}
\usepackage{latexsym,url}
\usepackage{enumerate}
\usepackage{algorithm}
\usepackage{algpseudocode}
\usepackage{graphics}
\usepackage{xparse} %
\usepackage{xspace}
\usepackage{multirow}
\usepackage{csvsimple}
\usepackage{balance}

\usepackage{
  tikz,
  pgfplots,
  pgfplotstable
}
\usepackage{hyperref}

\usepackage[inline]{enumitem}

\usetikzlibrary{
  shapes.geometric,
  arrows,
  external,
  pgfplots.groupplots,
  matrix,
  patterns,
  patterns.meta
}
\usepgfplotslibrary{fillbetween}

\newif\ifhaloonlyoutline
\haloonlyoutlinetrue %
\tikzset{
  haloonly/.cd,
  outline/.is if=haloonlyoutline,
  outline/.default=true,
}
\newif\ifzhalooutline
\zhalooutlinetrue %
\tikzset{
  zhalo/.cd,
  outline/.is if=zhalooutline,
  outline/.default=true,
}

\tikzset{
  pics/zhalo/.style args={w #1 h #2 inset #3}{
    code={
      \coordinate (-L) at (0,0);
      \coordinate (-U) at (#1,#2);

      \fill[yellow!85!white] (-L) rectangle (-U);

      \fill[
        pattern={Lines[angle=-45, distance=9pt, line width=4.5pt]},
        pattern color=orange!50!white
      ] (-L) rectangle (-U);

      \ifzhalooutline
      \draw[draw=black,densely dotted, thick] (-L) rectangle (-U);
      \fi

      \draw[fill=blue!20!white,draw=none] (#3,#3) rectangle ({#1-#3},{#2-#3});
    }
  },
  pics/haloonly/.style args={w #1 h #2}{
    code={
      \coordinate (-L) at (0,0);
      \coordinate (-U) at (#1,#2);

      \fill[yellow!85!white] (-L) rectangle (-U);

      \fill[
        pattern={Lines[angle=-45, distance=9pt, line width=4.5pt]},
        pattern color=orange!50!white
      ] (-L) rectangle (-U);

      \ifhaloonlyoutline
      \draw[draw=black,densely dotted, thick] (-L) rectangle (-U);
      \fi
    }
  }
}

\usepackage{algorithm}
\usepackage[noEnd=true,indLines=true,italicComments=false,spaceRequire=false]{algpseudocodex}

\usepackage{fancyhdr}
\pgfplotsset{compat=1.9}

\usepackage{mathtools}

\DeclareMathAlphabet{\mathcal}{OMS}{cmsy}{m}{n}

\DeclareGraphicsExtensions{%
    .png,.PNG,%
    .pdf,.PDF,%
    .jpg,.mps,.jpeg,.jbig2,.jb2,.JPG,.JPEG,.JBIG2,.JB2}

\usepackage{graphicx}

\theoremstyle{plain}
\newtheorem{theorem}{Theorem}[section]  %
\newtheorem{lemma}[theorem]{Lemma}
\newtheorem{proposition}[theorem]{Proposition}

\theoremstyle{definition}
\newtheorem{definition}[theorem]{Definition}
\newtheorem{example}[theorem]{Example}

\newcommand{\conf}{\ensuremath{\operatorname{conf}}\xspace}

\newcommand{\class}{\ensuremath{\operatorname{class}}\xspace}
\newcommand{\ReLU}{\ensuremath{\operatorname{ReLU}}\xspace}

\newcommand{\zonoLower}[1]{\ensuremath{\underline{#1}}}
\newcommand{\zonoUpper}[1]{\ensuremath{\overline{#1}}}

\newcommand{\affineForm}{\ensuremath{\mathfrak{z}}}%
\newcommand{\zonotope}{\ensuremath{\mathcal{Z}}}

\newcommand{\verydiff}{\textsc{VeryDiff}}
\newcommand{\ourTool}{\textsc{TwoSafe}}

\colorlet{isabelleCol}{cyan!80!black}
\makeatletter
\newtheoremstyle{lncsisacolor}%
  {0.5\baselineskip}{0.5\baselineskip}%
  {\itshape}%
  {}%
  {\color{isabelleCol}\bfseries}%
  {.}%
  { }%
  {}%
\makeatother
\theoremstyle{lncsisacolor}

\newtheorem{isaLemma}[lemma]{Lemma}

\theoremstyle{plain}

\begin{document}
\title{Differential Zonotopes for Verifying Global Robustness of DNNs}

\author{Anagha Athavale}
\orcid{0000-0002-1620-5700}
\affiliation{%
  \institution{TU Wien \\ AIT Austrian Institute of Technology}
  \city{Vienna}
  \country{Austria}
}

\email{anagha.athavale@tuwien.ac.at}

\author{Samuel Teuber}
\orcid{0000-0001-7945-9110}
\affiliation{%
  \institution{Karlsruhe Institute of Technology}
  \city{Karlsruhe}
  \country{Germany}
}
\email{teuber@kit.edu}

\author{Matteo Maffei}
\orcid{0000-0001-8061-1685}
\affiliation{%
  \institution{TU Wien}
  \city{Vienna}
  \country{Austria}
}
\email{matteo.maffei@tuwien.ac.at}

\author{Ezio Bartocci}
\orcid{0000-0002-8004-6601}
\affiliation{%
  \institution{TU Wien}
  \city{Vienna}
  \country{Austria}
}
\email{ezio.bartocci@tuwien.ac.at}

\author{Dejan Nickovic}
\orcid{0000-0001-5468-0396}
\affiliation{%
  \institution{AIT Austrian Institute of Technology}
  \city{Vienna}
  \country{Austria}
}
\email{dejan.nickovic@ait.ac.at}

\author{Georg Weissenbacher}
\orcid{0000-0002-0143-632X}
\affiliation{%
  \institution{TU Wien}
  \city{Vienna}
  \country{Austria}
}
\email{georg.weissenbacher@tuwien.ac.at}
\begin{abstract}
The robustness of deep neural networks (DNNs) is critical in security-sensitive applications, where small input perturbations should not alter model predictions. This property is commonly formalized as local or global robustness: the former considers perturbations around a single input, while the latter---strictly stronger---quantifies over all input pairs. While local robustness can be expressed as a safety property, global robustness is a 2-safety property, making it substantially more challenging to verify.

We present a novel static analysis technique for verifying the global robustness of DNNs. Our approach is based on differential halo zonotopes, a new abstract domain that extends zonotopes to jointly propagate pairs of perturbed inputs in lock-step while tightly bounding their divergence. In addition, we introduce a symmetric variant of confidence-based global robustness that disregards perturbations leading to differing but low-confidence predictions. This relaxation yields a practically meaningful notion of robustness that applies to a broader class of networks.

We implement our approach in a new tool, called \ourTool{}, and evaluate it on standard DNN verification benchmarks, including widely deployed models. Our results show that \ourTool{} significantly outperforms the state of the art in both precision and scalability, enabling the verification of networks an order of magnitude larger than those handled by prior techniques.
\end{abstract}

\maketitle    

\section{Introduction}
\label{intro}

Deep Neural Networks (DNNs)~\cite{goodfellow2016deep} are widely deployed today throughout society, enabling high levels of autonomy and efficiency in many domains, including transportation, healthcare and finance. For many of these applications, the accuracy of DNNs is an essential property, as wrong predictions can lead to significant economic loss or even harm to humans. Therefore, rigorous methods are often required to verify the behavior of DNNs, especially in security-critical settings.

\emph{Robustness}~\cite{seshia2018formal,katz2017reluplex,huang2017safety,gopinath2018deepsafe,athavale2024verifying} has emerged as a key property for safe and correct DNNs, requiring that small input perturbations do not alter model predictions. This is crucial since systems deployed in real-world environments must operate reliably under uncertainty and, possibly adversarial, variability in their inputs and operating conditions.
In the absence of robustness guarantees, even small input deviations (e.g., induced by sensor perturbations) can trigger disproportionate and unexpected changes in the output of a DNN, a consequence of the highly nonlinear structure of the common neural network architectures.
As a result, nominal correctness on typical inputs is insufficient to ensure correct behavior in practice.

In recent years, substantial effort has been invested in the development of methods and tools to formally verify the robustness of DNNs, with predominant approaches relying on the methods of satisfiability solving~\cite{katz2017reluplex,pulina2012challenging,huang2017safety,li2019analyzing} and static analysis~\cite{pulina2010abstraction, singh2019abstract,gehr2018ai2,10097028}.
Most existing methods focus on \emph{local} robustness verification, where robustness is checked only in the neighborhood of predefined input points. Although such guarantees are valuable, they only provide partial assurance and may fail to detect unsafe behaviors in unexplored regions of the input space.

 A recent line of work put forward the concept of confidence-based global robustness~\cite{athavale2024verifying},  which establishes robustness guarantees for all admissible inputs that lead to outputs with sufficiently high confidence. Global guarantees are a prerequisite for the true certification of DNNs and their reliable integration into security-critical applications. However, existing verification methods~\cite{athavale2024verifying} for confidence-based global robustness, based on specialized satisfiability solvers, still suffer from poor scalability.
To some extent, this does not come as a surprise, since local robustness is a safety property, whereas global robustness is a \emph{two-safety} property~\cite{clarkson2010hyperproperties}
which
is notoriously more expensive to verify since it requires confirming the relation~\cite{banerjee2024input,10.5555/3692070.3692181} between \emph{two} universally quantified execution traces, as opposed to a single one.  While local robustness tools scale to networks with thousands of neurons, 
global robustness verifiers are currently limited to networks with at most tens of neurons~\cite{athavale2024verifying}. 

Hence, there is at present an inherent tension between the strength of the robustness guarantees and the scalability of the associated verification tools. This poses the following research question: \emph{How can we improve the scalability of global robustness verification despite its inherent two-safety complexity?}

\paragraph{Contribution.}
We present a novel verification framework for confidence-based global robustness, which we show to outperform the state-of-the-art. %
In particular, we provide the following main contributions.
\begin{itemize} 
 
\item \emph{A differential zonotope-based static analysis technique.}
We develop an efficient verification technique for  confidence-based global robustness, which is based on the new concept of \emph{differential halo zonotopes}.
Our approach is inspired by recent work on differential zonotopes for DNN equivalence checking~\cite{teuber2025revisiting}, but requires non-trivial extension of these ideas to move from pairwise model comparison (the same inputs applied to two different DNNs) to global robustness guarantees (different inputs applied to the same DNN) and, in this setting,  to find the sweet spot between soundness and precision. 

\item \emph{A proof of soundness for differential zonotope-based  static analysis.} We formally prove the soundness of our verification technique, with the core results formalized in Isabelle (results marked in {\color{isabelleCol} cyan}).

\item \emph{A new definition of global robustness.} Our experiments indicate that confidence-based global robustness, as originally defined~\cite{athavale2024verifying}, can be overly restrictive and may fail to hold even for intuitively secure DNNs. In particular, the original definition—hereafter referred to as the \emph{asymmetric} variant—requires that no adversarial perturbation can transform a high-confidence classification into a different one, regardless of the confidence of the resulting classification.

Hence we introduce a weaker variant, which we call \emph{symmetric confidence-based global robustness}. This notion requires that two classifications coincide whenever both are of high confidence. As such, the symmetric definition is strictly less restrictive than the asymmetric one, while remaining meaningful in practical scenarios—for example, when low-confidence predictions are subject to manual inspection or deferred decision-making. Moreover, this definition admits an efficient verification procedure in our framework, as it can be encoded directly in our static analysis and checked using the same abstract domain. %

\item \emph{A tool accompanied with an extensive evaluation.} We implement our global robustness static analysis technique  on top of the propagation algorithm implemented in the \verydiff{} tool~\cite{teuber2025revisiting} for differential equivalence checking of DNNs.
We evaluate our approach on four case studies from the literature, including the LHC and HAR, which are actively deployed in practice. The LHC benchmark~\cite{duarte2018fast} is used for real-time particle classification on FPGA accelerators at CERN, and the HAR benchmark ~\cite{DBLP:conf/esann/AnguitaGOPR13} is used for human activity recognition on wearable devices, 
a technology increasingly deployed in healthcare and assisted living applications \cite{de2017mobile, serpush2022wearable, liu2021overview, bibbo2023human, sankaran2025future}.
In comparison to the State of the Art in global robustness verification~\cite{athavale2024verifying} scaling to 15 $\ReLU$ neurons, our approach scales to significantly wider (500 $\ReLU$s) and deeper DNNs (6 times 50 $\ReLU$s) while supporting larger input spaces (up to 561 dimensions).
\end{itemize}

\paragraph{Overview.}%
After \Cref{background} introduces the necessary background on zonotope-based static analysis and confidence-based robustness properties, \Cref{sec:property} introduces our new symmetric confidence-based property.
\Cref{sec:algorithm} then introduces our new analysis technique for the verification of confidence-based robustness and proves its soundness.
We evaluate our implementation in \Cref{sec:evaluation}.
\Cref{related} discusses related approaches and their similarities and differences.

\section{Background}
\label{background}

We denote vectors in bold (e.g., $\mathbf{v}$) and matrices using capital letters (e.g., $G$).
We consider the verification of piecewise linear, feed-forward DNNs. 

\subsection{Deep Neural Networks (DNNs)}

A DNN with input dimension $I \in \mathbb{N}$, input space $\mathcal{I} \subseteq \mathbb{R}^I$, output dimension $O \in \mathbb{N}$, 
and $L \in \mathbb{N}$ layers 
is a function
$f : \mathcal{I} \rightarrow \mathbb{R}^O$, mapping an input vector 
$\mathbf{x}^{(0)} \in \mathcal{I}$ to an output vector $\mathbf{x}^{(L)} \in \mathbb{R}^O$.
Each layer of the DNN consists of an affine transformation followed by a non-linear activation. 
More precisely, the $i$-th layer computes
\[
\hat{\mathbf{x}}^{(i)} = W^{(i)} \mathbf{x}^{(i-1)} + \mathbf{b}^{(i)},
\qquad
\mathbf{x}^{(i)} = h\!\left(\hat{\mathbf{x}}^{(i)}\right),
\]
where $W^{(i)}$ is a weight matrix, $\mathbf{b}^{(i)}$ a bias, and $h$ a non-linear activation function.

In this work, we focus on networks with ReLU activations, i.e.,
$h(\hat{\mathbf{x}}^{(i)})$ applies $\mathrm{ReLU}(\hat{\mathbf{x}}^{(i)}) = \max(0, \hat{\mathbf{x}}^{(i)})$ component-wise for all $1 \leq i \leq L$ and
$f(\mathbf{x})$ then denotes the composition of the network’s $L$ layers.

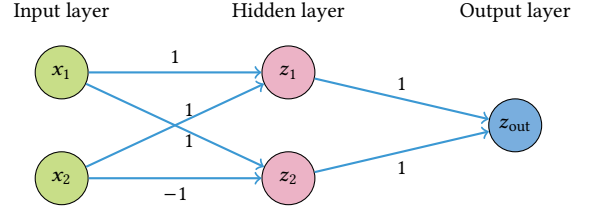
\begin{figure}[t]
\centering
\begin{tikzpicture}[
  x=1cm, y=1cm,
  annot/.style={font=\small},
  neuron/.style={circle, draw, minimum size=0.7cm,
                 inner sep=0pt, font=\small},
  input/.style={neuron, fill=yellow!60!green!60},
  hidden/.style={neuron, fill=purple!30!white},
  output/.style={neuron, fill=cyan!60!blue!50},
  edge/.style={->, thick, draw=cyan!60!blue!70},
  layer/.style={font=\small, align=center},
]

\node[layer] at (0,  1.5) {Input layer};
\node[layer] at (3,  1.5) {Hidden layer};
\node[layer] at (6,  1.5) {Output layer};

\node[input]  (x1)   at (0,  0.7) {$x_1$};
\node[input]  (x2)   at (0, -0.7) {$x_2$};
\node[hidden] (z1)   at (3,  0.7) {$z_1$};
\node[hidden] (z2)   at (3, -0.7) {$z_2$};
\node[output] (zout) at (6,  0.0) {$z_{\mathrm{out}}$};

\draw[edge] (x1) -- node[above, annot] {$1$} (z1);
\draw[edge] (x1) -- node[above right, annot] {$1$} (z2);
\draw[edge] (x2) -- node[below right, annot] {$1$} (z1);
\draw[edge] (x2) -- node[below, annot] {$-1$} (z2);
\draw[edge] (z1) -- node[above, annot] {$1$} (zout);
\draw[edge] (z2) -- node[below, annot] {$1$} (zout);

\end{tikzpicture}
\caption{Running example: a DNN with 2 inputs, 2 hidden $\ReLU$
neurons, and 1 output neuron.}
\label{fig:nn}
\end{figure}

\begin{example}[Forward Pass through a DNN]
\label{ex:dnn_forward}%
\Cref{fig:nn} shows a small DNN that we will use as a running example throughout the paper. It has 2 inputs $x_1, x_2 \in [0,1]$, 2 hidden $\ReLU$ neurons, and 1 output neuron. The weights are represented as 
\[
    W^{(1)} = \begin{pmatrix} 1 & 1 \\ 1 & -1 \end{pmatrix}, 
    \qquad 
    W^{(2)} = \begin{pmatrix} 1 & 1 \end{pmatrix},
\] and are also shown on each edge. The bias vectors are set to 0 and ignored for simplicity. 
\paragraph{Hidden layer.} The pre-activation value of $z_1$ 
is computed via the affine transformation:
\[
\hat{z}_1 = 1.0 \cdot x_1 + 1.0 \cdot x_2
\quad 
\hat{z}_2 = 1.0 \cdot x_1 - 1.0 \cdot x_2
\]
Applying the $\ReLU$ function: 
\[
z_1 = \ReLU(\hat{z}_1) = \max(0, x_1 + x_2)
\quad
z_2 =\ReLU(\hat{z}_2) = \max(0, x_1 - x_2) 
\]

\paragraph{Output layer.}
We compute the DNN's output as
\[
z_{out} = 1.0 \cdot z_1 + 1.0 \cdot z_2
\]
For input $\boldsymbol{x} = (0.5, 0.5)^\top$, this gives 
$z_1 = 1$, $z_2 = 0$, and $z_{out} = 1$.
\end{example}

We will focus on \emph{classification} DNNs which provide their classification as $\mathrm{class}\left(f\left(\boldsymbol{x}\right)\right) = \arg\max_{i\in\left[1,O\right]} f_i\left(\boldsymbol{x}\right)$ where $f_i$ is the $i$-th output component of $f$ and $\arg\max$ returns the component $i\in\left[1,O\right]$ for which $f_i\left(\boldsymbol{x}\right)$ is maximal.
Classifications of DNNs are often also interpreted as probability distributions via the function $\mathrm{softmax}_i\left(f\left(\boldsymbol{x}\right)\right) = \frac{e^{f_i\left(\boldsymbol{x}\right)}}{\sum_{j=1}^O e^{f_j\left(\boldsymbol{x}\right)}}$. 
In this case,
we compute the confidence of a DNN as
$\mathrm{conf}\left(f\left(\boldsymbol{x}\right)\right) = \mathrm{softmax}_{\mathrm{class}\left(f\left(x\right)\right)}\left(f\left(\boldsymbol{x}\right)\right)$.
In this work, we will focus on the case where $\mathrm{conf}\left(f\left(\boldsymbol{x}\right)\right) > 0.5$ prohibiting the case where two classes have the same confidence.

\subsection{Zonotopes}
A well-established abstraction technique for the verification of local properties in DNN  is based on \emph{zonotopes}~\cite{DBLP:conf/cav/GhorbalGP09,gehr2018ai2, Singh18}.
Zonotopes form an abstract domain enabling efficient propagation of input sets through DNNs composed of (piecewise) affine transformations.

\begin{definition}[Zonotope]
\label{def:zonotope}
A \emph{zonotope} with input and output dimension $n$ and $m$, respectively, consists of a \emph{generator matrix} $G \in \mathbb{R}^{m \times n}$ and a \emph{center} $\boldsymbol{c} \in \mathbb{R}^m$, which we denote as $\zonotope{} = \left(G,\boldsymbol{c}\right)$.
A zonotope's concretization encompasses
\[
\langle\left(G,\boldsymbol{c}\right)\rangle
= \left\{~
G\boldsymbol{\epsilon} + \boldsymbol{c} ~\middle|~
\boldsymbol{\epsilon} \in \left[-1,1\right]^n
~\right\} \subseteq \mathbb{R}^m.
\]
\end{definition}

\looseness=-1

Intuitively, a zonotope represents a set of points by starting at the center $\boldsymbol{c}$ and allowing movement along each \emph{generator} direction (column of $G$) by any amount in $[-1,1]$. Each $\epsilon_i \in [-1,1]$ independently controls 
how far we move along the $i$-th generator. The concretization 
$\langle(G,\boldsymbol{c})\rangle$ is then the set of \emph{all} points reachable this way.

Zonotopes are denoted by $\zonotope{}$.
We denote by $G_i$ the $i$-th row of $G$ and by $\zonotope{}_i$ the $i$-th component of $\zonotope{}$ which is determined by $G_i$ and $\boldsymbol{c}_i$. Each row $G_i$ together with $\boldsymbol{c}_i$ defines a scalar-valued function of $\boldsymbol{\epsilon}$ called an affine form $\affineForm{}=\left(\boldsymbol{g},c\right)\coloneqq \left(G_i,\boldsymbol{c}_i\right)$, representing one output dimension of the zonotope. A zonotope can thus be viewed is a sequence of $m$ affine forms with common generators. 
Given a zonotope and a vector $\boldsymbol{v}\in\left[-1,1\right]^d$, with $d \leq n$, we denote by $\zonotope{}\left(\boldsymbol{v}\right)$ the zonotope where we fix the first $d$ generators with the values in $\boldsymbol{v}$, i.e. we compute the constant $G_{(1:m),(1:d)} \boldsymbol{v}$.
By moving the resulting fixed value to to the zonotope's center, we can define the resulting zonotope as follows: $\zonotope\left(\boldsymbol{v}\right) \coloneqq \left( G_{(1:m),(d+1:n)},\enspace \left(G_{(1:m),(1:d)} \boldsymbol{v} + \boldsymbol{c}\right)\right)$.

\begin{example}
For a Zonotope $\zonotope{}=\left(G,\boldsymbol{0}\right)$ with $G = \left(\begin{array}{cc}
    1.0  & 0.0 \\
    2.0 & 1.0
\end{array}\right)$ and $\boldsymbol{v}=\left(-1\right)\in\mathbb{R}^1$, the zonotope $\zonotope\left(\boldsymbol{v}\right)$ is computed by fixing the first generator with $-1$ yielding $\left(\left(0\enspace 1\right)^T,\left(-1\enspace -2\right)^T\right)$.
\end{example} %
\noindent
By abuse of notation, for $d = n$ we also refer to $\zonotope{}\left(\boldsymbol{v}\right) \in \mathbb{R}^m$ as the vector obtained by setting $\boldsymbol{\epsilon}=\boldsymbol{v}$.

Zonotopes provide a good trade-off between efficiency and expressivity for the verification of DNNs.
The result of affine transformations can be computed exactly while preserving dependencies w.r.t. generators of the original zonotope: %
\begin{proposition}[Affine Zonotope Transformation~\cite{DBLP:conf/cav/GhorbalGP09,teuber2025revisiting}]
\label{prop:background:zono_affine}
    Let $\zonotope=\left(G,\mathbf{c}\right)$ be a zonotope and %
    $A\left(\mathbf{x}\right)=W\mathbf{x}+\mathbf{b}$ be an affine transformation.
    The zonotope $\hat{\zonotope}= W \zonotope + b \coloneqq \left(WG,W\mathbf{c}+\mathbf{b}\right)$ exactly describes $A$ applied to the points $x \in \langle \zonotope \rangle$, i.e.
    for all $d \leq n$ and $\mathbf{v} \in \left[-1,1\right]^d$:
    $
    \left\{ W\mathbf{x}+\mathbf{b} ~\middle|~ \mathbf{x} \in \langle \zonotope\left(\mathbf{v}\right) \rangle\right\} =
    \langle \hat{\zonotope}\left(\mathbf{v}\right) \rangle
    $
\end{proposition}
\noindent
Although expressive, zonotopes allow the efficient computation of bounds:
\begin{proposition}[Zonotope Output Bounds~\cite{DBLP:conf/cav/GhorbalGP09}]
\label{prop:background:zono_bounds}
    Consider some affine form $\affineForm{}=\left(\mathbf{g},c\right)$
    it holds for all $\mathbf{x} \in \left[-1,1\right]^n$ that:\\
    $
    \mathbf{g}\mathbf{x}+c \in
    \left[\zonoLower{\affineForm{}},\zonoUpper{\affineForm{}}\right] \coloneqq
    \left[ c-\sum_{i=0}^n \left|\mathbf{g}_{i}\right|,c+\sum_{i=0}^n \left|\mathbf{g}_{i}\right|  \right]
    $
\end{proposition}%

There exist effective abstract transformers to over-approximate the behavior of $\ReLU$~\cite{Singh18}:
If the output of a $\ReLU$ node is \emph{stable} (i.e. its input is either entirely positive or entirely negative), the $\ReLU$ behaves linearly and we can \emph{exactly} represent its behavior with zonotopes.
On the other hand, if a \ReLU is \emph{unstable} (i.e. its input range crosses includes $0$), zonotopes (as linear transformations of a $[-1,1]$ box) cannot exactly represent its behavior of $\ReLU$ and we must resort to approximation instead.

\begin{wrapfigure}[8]{r}{0.25\textwidth}
\centering
\begin{tikzpicture}
  \begin{axis}[
    axis lines = middle,
    xlabel = {$x$},
    ylabel = {$y=\ReLU\left(x\right)$},
    ylabel style = {anchor=south},
    y=3.75cm,
    x=3.5cm,
    xmin=-0.5, xmax=0.7,
    ymin=-0.15, ymax=0.25,
    xtick distance=5.0,
    ytick distance=5.0]
  \addplot [name path = relu1,
    domain = -0.45:0,
    samples = 20,red] {0};
  \addplot [name path = relu2,
    domain = 0:0.45,
    samples = 20,red] {0.5*x};
  \addplot [name path = zonoUp,
    domain = -0.45:0.45,
    samples = 20, teal] {0.25*x+0.25*0.45};
  \addplot [name path = zonoDown,
    domain = -0.45:0.45,
    samples = 20, teal] {0.25*x)};
  \addplot [teal!30] fill between [of = zonoDown and zonoUp, soft clip={domain=-0.9:0.9}];
  \draw [<->] (axis cs:0.45,0.1125) -- node[anchor=west] {$\boldsymbol{\epsilon}_{\text{new}}$} (axis cs:0.45,0.225);
  \end{axis}
\end{tikzpicture}
\caption{$\ReLU$ approximation}
\label{fig:background:zono_relu_approx}
\end{wrapfigure}%
\noindent
To this end, the zonotope's generator dimension $n$ is increased once per $\ReLU$ neuron and the additional generator $\boldsymbol{\epsilon}_{\text{new}}$ is used to represent the approximation error incurred through representing $\ReLU$ via a linear approximation. %
\Cref{fig:background:zono_relu_approx} illustrates how 
an input zonotope is mapped through $\ReLU$ (red function) to produce 
an over-approximating output zonotope (cyan region). The additional 
generator $\boldsymbol{\epsilon}_{\text{new}}$ captures the 
approximation error.
Overall, we add one additional generator per unstable $\ReLU$ neuron in the NN.

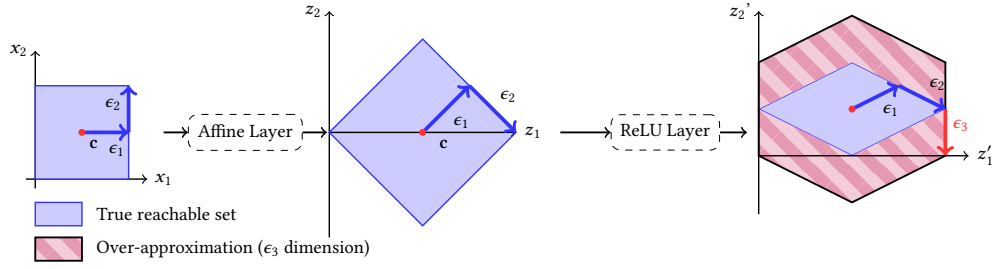
\begin{figure*}[t]
\centering
\resizebox{0.75\linewidth}{!}{
\begin{tikzpicture}[
  x=0.35cm, y=0.35cm,
  annot/.style={font=\small},
  arrow/.style={->, thick},
  netTransform/.style={align=center, draw, rounded corners,
                        dashed, font=\small, inner sep=4pt},
]

  \draw[->, line width=0.6pt] (-0.4,0) -- (4.8,0)
      node[anchor=west, annot] {$x_1$};
  \draw[->, line width=0.6pt] (0,-0.4) -- (0,5.5)
      node[left, annot] {$x_2$};
  \fill[blue!20!white] (0,0) rectangle (4,4);
  \draw[blue!80!white, line width=0.6pt] (0,0) rectangle (4,4);
  \draw[->, blue!80!white,line width=1.6pt] (2,2) -- (4,2);
  \node[annot, below, black] at (3.5,1.8) {$\epsilon_1$};
  \draw[->, blue!80!white,line width=1.6pt] (4,2) -- (4,4);
  \node[annot, left, black] at (4,3.2) {$\epsilon_2$};
  \fill[red!80!white] (2,2) circle (0.15);

\node[netTransform] (aff) at (9,2)
     {Affine Layer};
\draw[arrow] (5.5,2) -- (aff.west);
\draw[arrow] (aff.east) -- (12.5,2);
\node[annot, below right, black] at (2,2) {$\mathbf{c}$};

  \fill[blue!20!white] (12.6,2) -- (16.6,6) -- (20.6,2) -- (16.6,-2) -- cycle;
  \draw[->, line width=0.6pt] (12.6,2) -- (20.6,2)
      node[anchor=west, annot,xshift=0.5] {$z_1$};
  \draw[->, line width=0.6pt] (12.6,-2.5) -- (12.6,7.2)
      node[left, annot] {$z_2$};
  \draw[blue!80!white, line width=0.6pt]
       (12.6,2) -- (16.6,6) -- (20.6,2) -- (16.6,-2) -- cycle;
  \draw[->, blue!80!white,line width=1.6pt] (16.6,2) -- (18.6,4);
  \node[annot, below, black,anchor=north west] at (17.6,3) {$\epsilon_1$};
  \draw[->, blue!80!white,line width=1.6pt] (18.6,4) -- (20.6,2);
  \node[annot, left, black, anchor=south west] at (19.6,3) {$\epsilon_2$};
  \fill[red!80!white] (16.6,2) circle (0.15);

\node[netTransform] (relu) at (27,2)
    {$\ReLU$ Layer};
\draw[arrow] (22.5,2) -- (relu.west);
\draw[arrow] (relu.east) -- (30.5,2);
\node[annot, below right, black] at (17,2) {$\mathbf{c}$};

\begin{scope}[shift={(0,-1)}]
  \fill[blue!20!white] (31,4) -- (35,6) -- (39,4) -- (35,2) -- cycle;
  \draw[blue!80!white, line width=0.6pt]
       (31,4) -- (35,6) -- (39,4) -- (35,2) -- cycle;

  \fill[purple!20!white]
    (31,4) -- (31,2) -- (35,0) -- (39,2) -- (39,4) -- (35,2) -- cycle;
  \fill[pattern={Lines[angle=-45, distance=9pt, line width=4.5pt]},
        pattern color=purple!40!white]
    (31,4) -- (31,2) -- (35,0) -- (39,2) -- (39,4) -- (35,2) -- cycle;

  \fill[purple!20!white]
    (31,4) -- (31,6) -- (35,8) -- (39,6) -- (39,4) -- (35,6) -- cycle;
  \fill[pattern={Lines[angle=-45, distance=9pt, line width=4.5pt]},
        pattern color=purple!40!white]
    (31,4) -- (31,6) -- (35,8) -- (39,6) -- (39,4) -- (35,6) -- cycle;

  \draw[black, thick]
  (31,2) -- (35,0) -- (39,2) -- (39,6) -- (35,8) -- (31,6) -- cycle;

  \draw[->, blue!80!white,line width=1.6pt] (35,4) -- (37,5);
  \node[annot, below, black,anchor=north west] at (36,4.5) {$\epsilon_1$};
  \draw[->, blue!80!white,line width=1.6pt] (37,5) -- (39,4);
  \node[annot, left, black, anchor=south west] at (38,4.5) {$\epsilon_2$};
  \draw[->, red!80!white,line width=1.6pt] (39,4) -- (39,2);
  \node[annot, left, red!80!white, anchor=south west] at (39,2.75) {$\epsilon_3$};
  \fill[red!80!white] (35,4) circle (0.15);
  \draw[->, line width=0.6pt] (31,2) -- (40,2)
      node[anchor=west, annot,xshift=0.5] {$z_1'$};
  \draw[->, line width=0.6pt] (31,-0.5) -- (31,8.2)
      node[left, annot] {$z_2$'};
  \end{scope}

  \fill[blue!20!white] (0,-2) rectangle (2,-1);
  \draw[blue!80!white, line width=0.6pt] (0,-2) rectangle (2,-1);
  \node[annot, anchor=west] at (2.3,-1.5) {True reachable set};
  \fill[purple!20!white] (0,-3.5) rectangle (2,-2.5);
  \fill[pattern={Lines[angle=-45, distance=9pt, line width=4.5pt]},
        pattern color=purple!40!white]
      (0,-3.5) rectangle (2,-2.5);
  \draw[black, thick] (0,-3.5) rectangle (2,-2.5);
  \node[annot, anchor=west] at (2.3,-3)
      {Over-approximation ($\epsilon_3$ dimension)};

\end{tikzpicture}
}
\caption{Zonotope propagation through the running example
(\Cref{ex:zonotope_propagation}): the input zonotope is transformed by the affine layer and then over-approximated through the $\ReLU$
layer. The blue face shows the true reachable set; the pink hatched faces show the over-approximation 
in the extra $\epsilon_3$ dimension, introduced solely because neuron $z_2$ is unstable. The stable neuron $z_1$ requires no new generator.
The arrows illustrate how the generators span the (over-approximate) zonotope.
}
\label{fig:zono_propagation}
\end{figure*}

\begin{example}[Zonotope Propagation through a DNN]
\label{ex:zonotope_propagation}
Consider the DNN from \Cref{fig:nn} with weights: 
\[
W^{(1)} = \begin{pmatrix} 1 & 1 \\ 1 & -1 \end{pmatrix}, 
 \quad
W^{(2)} = \begin{pmatrix} 1 & 1 \end{pmatrix}
\]

We verify the network over the input region $x_1, x_2 \in [0,1]$, 
which we represent as a zonotope with center 
$\boldsymbol{c} = (0.5,0.5)^\top$ and generator matrix 
$G = \begin{pmatrix} 0.5 & 0 \\ 0 & 0.5 \end{pmatrix}$,
so that $x_1 = 0.5\epsilon_1 + 0.5$ and $x_2 = 0.5\epsilon_2 + 0.5$ 
for $\epsilon_1, \epsilon_2 \in [-1,1]$.

\paragraph{Affine layer.}
Applying Proposition~\ref{prop:background:zono_affine}, the 
pre-activation zonotope is:
\[
\hat{G} = W^{(1)} G = 
\begin{pmatrix} 0.5 & 0.5 \\ 0.5 & -0.5 \end{pmatrix}, 
\quad
\hat{\boldsymbol{c}} = W^{(1)}\boldsymbol{c} = 
\begin{pmatrix} 1 \\ 0 \end{pmatrix}.
\]
The two pre-activation affine forms are:
\begin{align*}
z_1 &= 0.5\epsilon_1 + 0.5\epsilon_2 + 1 \in [0, 2]\\
z_2 &= 0.5\epsilon_1 - 0.5\epsilon_2 + 0 \in [-1, 1]
\end{align*}

\paragraph{$\ReLU$ layer.}
\emph{Neuron 1} ($z_1$): since $l_1 = 0 \geq 0$, the neuron is stable
(always active) and $\ReLU(z_1) = z_1$ exactly; no new generator 
is needed. %

\emph{Neuron 2} ($z_2$): since $l_2 = -1 < 0 < 1 = u_2$, the 
neuron is unstable. We introduce a new generator $\epsilon_3$ 
to over-approximate $\ReLU(z_2)$ \cite{Singh18}:
\[
z_2' = 0.25\epsilon_1 - 0.25\epsilon_2 + 0.5\epsilon_3 + 0.5
\]
The generator matrix now has three columns 
$(\epsilon_1, \epsilon_2, \epsilon_3)$, one per input 
dimension and one for the $\ReLU$ approximation error. 

\paragraph{Output layer.}
Applying the output weights $W^{(2)} = (1,1)$:
\[
z_{\text{out}} = z_1 + z_2' = 
0.75\epsilon_1 + 0.25\epsilon_2 + 0.5\epsilon_3 + 1.5
\]
By Proposition~\ref{prop:background:zono_bounds}, the output is then bounded by:
\[
z_{\text{out}} \in \left[1.5 - (0.75 + 0.25 + 0.5),\; 
1.5 + (0.75 + 0.25 + 0.5)\right] = [0, 3]
\]
This is an over-approximation of the true output range due to the 
$\ReLU$ approximation at neuron 2. \Cref{fig:zono_propagation} visualizes this propagation: the input 
zonotope (a square in the $x_1$-$x_2$ plane) 
is transformed by the affine layer into the diamond-shaped zonotope, reflecting the rotation and shearing induced by $W^{(1)}$. After the $\ReLU$ layer, the stable neuron $z_1$ is 
propagated exactly, while the unstable neuron $z_2$ introduces a new generator $\epsilon_3$, represented as the extra dimension of the resulting zonotope. The blue face shows the true reachable set; the pink hatched faces show the over-approximation introduced by the new generator $\epsilon_3$.
\end{example}

\subsection{Zonotopes for Equivalence Verification.}
\looseness=-1
Teuber \emph{et al.}~\cite{teuber2025revisiting} propose a zonotope-based abstract domain to prove equivalence properties across two DNNs using the differential verification paradigm~\cite{paulsen_reludiff_2020,paulsen_neurodiff_2020}.
Given two DNNs of the same structure (modulo weight changes) and an input region of interest, a zonotope-based reachability analysis is performed for both individual DNNs.
In addition, a \emph{differential zonotope} $\zonotope^\Delta$ is propagated in lock-step with the individual analyses.
$\zonotope^\Delta$ then yields a bound on the \emph{difference} between executions which is tighter than the bounds obtained via individual analyses~\cite{teuber2025revisiting}.

We can then reason about differences in DNN behavior by analyzing the zonotopes propagated through the individual DNNs ($\zonotope'$/$\zonotope''$) and the zonotope bounding the difference between executions ($\zonotope^\Delta$) in combination.
Any tuple of points $\left(\boldsymbol{x},\boldsymbol{y}\right)$ is contained in $\langle\left(\zonotope',\zonotope'',\zonotope^\Delta\right)\rangle$
iff they are contained in the individual zonotopes ($\zonotope',\zonotope''$) and their difference $\boldsymbol{x}-\boldsymbol{y}$ is contained in $\zonotope^\Delta$:

\begin{definition}[Concretization of Differential Zonotopes~\cite{teuber2025revisiting}]
\label{def:diff_affine_form_containment}
For two individual zonotopes $\zonotope',\zonotope''$ and a differential zonotope $\zonotope^\Delta$ we define the concretization of $\left(\zonotope',\zonotope'',\zonotope^\Delta\right)$ as follows:
\begin{align*}
\langle\left(\zonotope',\zonotope'',\zonotope^\Delta\right)\rangle = 
\left\{
\left(\boldsymbol{x},\boldsymbol{y}\right) ~\middle|~ 
\begin{aligned}
\exists \boldsymbol{\epsilon}\in\left[-1,1\right]^n \quad
\zonotope'\left(\boldsymbol{\epsilon}\right) &= \boldsymbol{x}\enspace\land\\
\zonotope''\left(\boldsymbol{\epsilon}\right) &= \boldsymbol{y}\enspace\land\\
\zonotope^\Delta\left(\boldsymbol{\epsilon}\right) &= \left(\boldsymbol{x}-\boldsymbol{y}\right)
\end{aligned}
\right\}
\end{align*}
\end{definition}%

A challenge in this approach lies in keeping track of generators added due to necessary approximations of unstable ReLU neurons in $\zonotope'$, $\zonotope''$ or $\zonotope^\Delta$. 
To this end, one can represent zonotopes as consisting of multiple generator blocks (e.g. denoting $\zonotope'=\left(E',A',\boldsymbol{c}\right)$ for $E'\in\mathbb{R}^{m\times n_1}$, $A'\in\mathbb{R}^{m\times n_2}$). %
Generator blocks help to distinguish between generator columns which were part of the initial input encoding (\emph{exact} generators) and generators added later for $\ReLU$ approximations (\emph{approximate} generators). This distinction is important because $\zonotope'$, $\zonotope''$, and $\zonotope^\Delta$ must share the same generator columns to maintain the \emph{correlation} between the two executions.
Without this alignment, the bound on the difference $\zonotope^\Delta$ would be no tighter than simply subtracting the individual zonotopes $\zonotope' - \zonotope''$.
By concatenating the generator blocks and inserting zero-blocks, one obtains standardized zonotopes aligned between $\zonotope',\zonotope'',\zonotope^\Delta$. For simplicity, we omit generator blocks in our presentation. %

Differential zonotopes achieve their improved capabilities in bounding the difference between DNN executions through specialized abstract transformers which update $\zonotope^\Delta$ based on differences between the DNNs as well as differences accumulated in $\zonotope^\Delta$ in previous layers.
For example, for affine transformations, the resulting $\zonotope^\Delta$ can be computed as follows:
\begin{proposition}[Affine Differential Zonotope Transformation~\cite{teuber2025revisiting}]
Given two affine transformations $W_1 \boldsymbol{x} + \boldsymbol{b_1}$ and $W_2 \boldsymbol{y} + \boldsymbol{b_2}$ and
zonotopes $\zonotope'$, $\zonotope''$ and $\zonotope^\Delta$, let $\hat{\zonotope}',\hat{\zonotope}'',\hat{\zonotope}^\Delta$ be as follows:
\begin{align*}
    \hat{\zonotope}' &=\enspace W_1 \zonotope' + \boldsymbol{b_1}, &
    \hat{\zonotope}'' &=\enspace W_2 \zonotope' + \boldsymbol{b_2},\\
    \hat{\zonotope}^\Delta &=\enspace \rlap{$W_1 \zonotope^\Delta  + \left(W_1 - W_2\right) \zonotope'' + \left(\boldsymbol{b_1} - \boldsymbol{b_2}\right)$}.
\end{align*}
Then for $d \leq n$ and $\boldsymbol{v} \in \left[-1,1\right]^d$
with\\
$(\boldsymbol{x},\boldsymbol{y}) \in \langle (\zonotope{}'\left(\boldsymbol{v}\right),\zonotope{}''\left(\boldsymbol{v}\right),\zonotope{}^\Delta\left(\boldsymbol{v}\right)) \rangle$ 
 we get:
\[
\left(\left(W_1 \boldsymbol{x} + \boldsymbol{b_1}\right),\left(W_2 \boldsymbol{y} + \boldsymbol{b_2}\right)\right)
\in 
\langle (\hat{\zonotope}_1\left(\boldsymbol{v}\right),\hat{\zonotope}_2\left(\boldsymbol{v}\right),\hat{\zonotope}^\Delta\left(\boldsymbol{v}\right)) \rangle
\]
\end{proposition}
A similar (though in this case approximate) transformation can be constructed for the difference between two $\ReLU$ nodes by rewriting $\ReLU\left(x\right)-\ReLU\left(y\right)$ to $\ReLU\left(x\right) - \ReLU\left(x-\Delta\right)$~\cite{teuber2025revisiting}.

\begin{example}[Differential Zonotope Propagation through a DNN]
Consider again the DNN from \Cref{fig:nn}.
For sake of simplicity we verify the equivalence of two copies of this network over the input region $x_1, x_2 \in [0, 1]$, represented as the zonotope from \Cref{ex:zonotope_propagation}. 
Since we are verifying equivalence, both networks receive the same 
input, so we initialize:
\[
    \zonotope' = \zonotope'' = (G, \boldsymbol{c}), 
    \quad 
    \zonotope^\Delta = (\mathbf{0}, \mathbf{0})
\]
Since both networks share the same weights and biases (e.g.\ $W_1-W_2=\mathbf{0}$), the affine update from \Cref{prop:background:zono_affine} gives $\hat{\zonotope}^\Delta = (\mathbf{0}, \mathbf{0})$, and the individual zonotopes transform as in \Cref{ex:zonotope_propagation}.
At the unstable neuron $z_2$, both $\zonotope'$ and $\zonotope''$ \emph{independently} introduce approximate generators $\epsilon_3'$ and $\epsilon_3''$.
Naively subtracting $\zonotope'$ and $\zonotope''$ then yields the affine form $0.5 \epsilon_3' - 0.5\epsilon_3''$ which leads to a difference bound of $\left[-1,1\right]$.
In contrast, the transformation of the differential zonotope $\zonotope^\Delta$ accounts for the previously accumulated difference $\left[0,0\right]$ yielding the output bound $[0,0]$ which proves equivalence.
This difference in precision underscores the effectiveness of differential verification for hyper-properties where naive zonotope propagation suffers from accumulating over-approximation errors.

The initialization $\zonotope' = \zonotope''$ and $\zonotope^\Delta = \mathbf{0}$ is specific to equivalence verification. As we show in \Cref{sec:algorithm}, naively reusing this initialization for global robustness verification, where the two executions receive 
\emph{different} inputs, leads to unsound or imprecise results --- 
motivating our novel differential halo zonotope encoding.
\end{example}

The differential zonotope abstract domain is implemented in the tool \verydiff{}~\cite{teuber2025revisiting} which, however, only supports equivalence verification, i.e. the equivalence of two DNNs where the second DNN was either created through weight modifications~\cite{teuber2025revisiting} (e.g., via pruning) or through polynomial approximation of activation functions~\cite{kern2025certified} (for fully homomorphic encryption).
Prior work has shown that differential verification, especially in classification settings, is very sensitive to the considered property, i.e., not all relational properties benefit equally from differential verification~\cite[Sec. 6]{teuber2025revisiting}.

\subsection{Asymmetric Confidence-Based Global Robustness}
\label{subsec:background:confidence}
A widely studied property of DNNs is \emph{local robustness}, which requires that small perturbations of a fixed input do not change the network's classification. Formally, a DNN $f$ is locally robust at an input $\boldsymbol{x}_0$ with tolerance $\boldsymbol{\varepsilon}$ iff
\[
    \forall \boldsymbol{x} \in \mathcal{I}.\quad 
    |\boldsymbol{x}_0 - \boldsymbol{x}| \leq \boldsymbol{\varepsilon} 
    \rightarrow \class(f(\boldsymbol{x}_0)) = \class(f(\boldsymbol{x}))
\]

While local robustness provides meaningful guarantees around specific inputs, it leaves the rest of the input space unconstrained. A natural attempt to generalize this to a \emph{global} property is to require classification consistency for \emph{every} pair of nearby inputs:
\[
    \forall \boldsymbol{x}, \boldsymbol{y} \in \mathcal{I}.\quad 
    |\boldsymbol{x} - \boldsymbol{y}| \leq \boldsymbol{\varepsilon} 
    \rightarrow \class(f(\boldsymbol{x})) = \class(f(\boldsymbol{y}))
\]
However, this naive lifting is unsatisfiable for any non-trivial classifier: along any decision boundary, two arbitrarily close inputs are mapped to different classes by definition. 

To overcome this, Athavale et al.~\cite{athavale2024verifying} introduce \emph{confidence-based global robustness}, which restricts the global requirement to inputs classified with high confidence \cite{kull2019beyond, guo2017calibration,ao2023two}. The intuition is that decision boundaries are typically associated with low confidence, so excluding low-confidence inputs makes a global property attainable. Formally, we define confidence-based global robustness as follows:

\begin{definition}[Asymmetric Confidence-Based Global Robustness~\cite{athavale2024verifying}] \label{def:conf_rob_old}
A DNN $f$ satisfies asymmetric global robustness for confidence $\tau > 0$ and tolerance $\boldsymbol{\varepsilon}$ iff
\begin{displaymath}
\begin{aligned}
\forall \boldsymbol{x}, \boldsymbol{y}\in\mathcal{I}\,.\quad
       |\boldsymbol{x} - \boldsymbol{y}| \leq \boldsymbol{\varepsilon} \wedge \conf(f(\boldsymbol{x})) > \tau \enspace \\ \rightarrow \enspace \class(f(\boldsymbol{x})) = \class(f(\boldsymbol{y}))
      \end{aligned}
\end{displaymath}
\end{definition}

Asymmetric confidence-based global robustness~\cite{athavale2024verifying}, intuitively, ensures that for any data instance $\boldsymbol{x}$ classified with high confidence $\tau$, no other data instance, $\boldsymbol{y}$, that is $\varepsilon$ close to $\boldsymbol{x}$ shall be classified to a different class. The asymmetry is reflected in the confidence requirement being placed only on $\boldsymbol{x}$ and not on $\boldsymbol{y}$.

\section{Symmetric Confidence-Based Global Robustness} 
\label{sec:property}

\looseness=-1
\Cref{subsec:background:confidence} presented confidence-based robustness as a property with the potential to provide strong robustness guarantees on the entire input space of a DNN.
We note that this property stipulates a  
strong requirement, as it  precludes the possibility of low-confidence adversarial examples.
We argue that in several situations, it is meaningful to \emph{allow} the network to classify $\boldsymbol{y}$ to a different class, if it does so with a low confidence score (e.g., when a low confidence score triggers a warning for manual intervention). %

To this end, we introduce a more permissive,
 \textit{symmetric} confidence\nobreakdash-based global robustness property (\Cref{def:conf_rob_new}) that only constrains pairs of inputs that are
both classified with high confidence to different classes.
The difference between the asymmetric and the symmetric confidence-based robustness
properties is illustrated in Figure~\ref{fig:symmetric-motivation}.
The pair of pink points are marked safe by both asymmetric and symmetric properties, since they are classified to the same class with high confidence scores.
The orange pair is marked unsafe by both properties since the points are classified to different classes with high confidence scores.
The blue pair, however, highlights the limitation of asymmetric confidence-based global robustness, because one blue point is classified as \class 3 with high confidence, while the other blue point is classified to \class 1 with low confidence.
The asymmetric property flags this as a violation, whereas the symmetric property considers it safe, since the
disagreement involves a low-confidence prediction that would not be trusted in
practice. 

\begin{figure}
\centering
\includegraphics[width=0.25\textwidth]{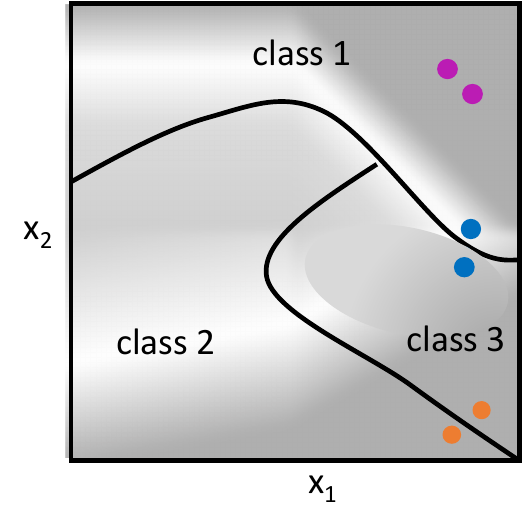}
\caption{Symmetric vs asymmetric confidence-based global robustness.
$x_1$ and $x_2$ denote points in the continuous input space. Grayscale shading indicates confidence, where darker regions correspond to higher confidence. The pink pair satisfies both properties. The blue pair violates the asymmetric property but satisfies the symmetric property due to low confidence at the perturbed point. The orange pair
violates both properties.
 }
\label{fig:symmetric-motivation}
\end{figure}

The figure highlights that while small perturbations may cause a change in the predicted class, such changes are often accompanied by a significant drop in confidence.
Treating all such cases as violations, as required by the asymmetric definition, can be overly restrictive.
Therefore, we propose the following symmetric property, which is more permissive in this regard:

\begin{definition}[Symmetric Confidence-Based Global Robustness]\label{def:conf_rob_new}
A DNN $f$ satisfies symmetric global robustness for confidence $\tau > 0$ and tolerance $\boldsymbol{\varepsilon}$ iff 
\begin{displaymath}
\begin{split}
    \forall \boldsymbol{x}, \boldsymbol{y}\in\mathcal{I}\,.\quad
      | \boldsymbol{x} - \boldsymbol{y}| \leq \boldsymbol{\varepsilon} \wedge \conf(f(\boldsymbol{x})) > \tau \wedge \conf(f(\boldsymbol{y})) > \tau \\
      \rightarrow \class(f(\boldsymbol{x})) = \class(f(\boldsymbol{y}))
\end{split}
\end{displaymath}
\end{definition}

\noindent
\looseness=-1
This property requires that for any two $\varepsilon$-close inputs that are both classified with high confidence, the DNN maps the inputs to the same class.
For instance, consider a DNN classifying images as \emph{cats} and \emph{dogs}.
Assume an image of a dog is classified to class \textit{dog} with 95\% confidence, but tweaking it by a few pixels makes it classify as \textit{cat} with confidence 60\%.
\Cref{def:conf_rob_old} does not allow this and marks the DNN as unsafe.
However, this is a very strong requirement on the network, which is unlikely to be necessary in practice.
If we knew the DNN was robust for inputs with confidence $\geq 80\%$, it would be reasonable to \emph{reject} any classifications with confidence $<80\%$ -- a common risk limitation strategy for DNNs~\cite{DBLP:conf/nips/GeifmanE17}.
Doing so would in particular reject the result for the low-confidence image above, which was wrongly classified as cat.
Hence, we argue adversarial examples may be considered permissible if they only lead to a low confidence classification.
This observation is precisely covered by \Cref{def:conf_rob_new}:
Symmetric confidence-based global robustness does \emph{not} flag this case as a counterexample.
It would, however, be unacceptable if the classification decision changes to cat with high confidence (e.g. with $\geq 80\%$ confidence), as we then would likely accept this classification as correct.
Our new property imposes a weaker requirement on the DNN and
hence can be fulfilled by a broader class of DNNs than the asymmetric version.
\begin{lemma}
\label{cor:asymm-implies-symm}
If a DNN satisfies asymmetric confidence-based global robustness for parameters
$(\epsilon, \tau)$, then it also satisfies symmetric
confidence\nobreakdash-based global robustness for parameters $(\epsilon, \tau')$,
for any $\tau' \ge \tau$.

Conversely, for $\tau>0.5$, a DNN satisfying symmetric confidence-based global robustness for parameters $(\epsilon,\tau)$ does not necessarily satisfy asymmetric confidence-based global robustness for $(\epsilon,\tau)$.
\end{lemma}
\begin{proof}
\label{proof:asymm-implies-symm}
Assume the asymmetric confidence-based robustness property holds for parameters $\left(\boldsymbol{\varepsilon},\tau\right)$ for some DNN $f$.
This requires that for any input $\boldsymbol{x}$ classified to class $\mathsf{A}$ with confidence $\conf(f(\boldsymbol{x})) > \tau$, every input $\boldsymbol{y}$ satisfying
$| \boldsymbol{x} - \boldsymbol{y} | \le \boldsymbol{\varepsilon}$ is also classified to $\mathsf{A}$.
Consider a pair of inputs $\boldsymbol{x}, \boldsymbol{y}$ such that
$| \boldsymbol{x} - \boldsymbol{y} | \le \boldsymbol{\varepsilon}$,
$\conf(f(\boldsymbol{x})) > \tau' \geq \tau$, and $\conf(f(\boldsymbol{y})) > \tau' \geq \tau$.
Then we know that $\conf(f(\boldsymbol{x})) > \tau$ as $\tau' \geq \tau$.
From our assumption, it follows that $\boldsymbol{y}$ must be classified to the same class as $\boldsymbol{x}$.
Therefore, the symmetric confidence-based global robustness property holds for threshold
$\tau'$.

To show that symmetric confidence-based global robustness for $(\epsilon,\tau)$ does not necessarily imply asymmetric robustness for $(\epsilon,\tau)$, consider the DNN $f\left(x\right)=\left( 1\quad 0\right)^{\mathrm{T}} x$.
The DNN $f$ has confidence $\conf\left(f\left(x\right)\right)\geq 0.75$ for $\class\left(f\left(x\right)\right)=0$ with $x\in\left[\ln\left(0.75/0.25\right),\infty\right)$ and for $\class\left(f\left(x\right)\right)=1$ with $x\in\left(-\infty,-\ln\left(0.75/0.25\right)\right]$.
Moreover, $f$ classifies $x$ to class $0$ iff $x \geq 0$.
Hence, for $\tau=0.75$ and $\varepsilon=1.5$ the DNN satisfies symmetric global robustness (differing high-confidence classifications are $2\ln\left(0.75/0.25\right)$ apart), but not asymmetric global confidence ($x=1.1\geq\ln\left(0.75/0.25\right)$ is classified as $0$ with high confidence but $-0.4$ (which is $\varepsilon$-close to $x$) is classified as $1$ with low confidence).
\end{proof}
\looseness=-1
Verifying global robustness requires
reasoning about two executions of a DNN over large input regions.
Below, we introduce a zonotope-based approach leveraging differential verification to verify confidence-based robustness properties.

\section{Robustness Verification with Differential Zonotopes}
\label{sec:algorithm}
In prior work, differential zonotope-based verification has been used to verify the \emph{equivalence} of two DNNs w.r.t. a given input space.
Equivalence properties assert that the DNNs' outputs are equivalent when provided with the same inputs.
Hence, the inputs of the two DNNs are assumed \emph{equal} across the two executions while the DNNs themselves \emph{differ}.
Consequently, the algorithm from Teuber \emph{et al.}~\cite{teuber2025revisiting} by construction initializes the input zonotopes as: $\zonotope'=\zonotope''$ and $\zonotope^\Delta=0$

In contrast, global robustness verification attempts to prove a property about two executions of the \emph{same} DNN w.r.t. \emph{differing} inputs.
This raises the question of how this relational property can be encoded for differential zonotope verification in a \emph{sound} and \emph{precise} manner -- as we show in \Cref{subsec:algorithm:failed}, straightforward adaptations of the approach by Teuber \emph{et al.}~\cite{teuber2025revisiting} lead to  unsound or  imprecise results.

Below we assume a box-shaped input space.
We describe possible encodings w.r.t. 1-dimensional affine forms $\affineForm{}$ that can be independently combined to describe multidimensional inputs.
We consider an input interval $\left[l,u\right]$ and perturbation $\varepsilon$ and will also use the center $c=\left(u+l\right)/2$ and the range $r=\left(u-l\right)/2$.
Our objective will be to construct affine forms $\affineForm{}',\affineForm{}'',\affineForm{}^\Delta$ which ensure that $\left(x,y\right)\in\langle(\affineForm{}',\affineForm{}'',\affineForm{}^\Delta)\rangle$ iff $x,y\in\left[l,u\right]$ and $\left(x-y\right)\in\left[-\varepsilon,\varepsilon\right]$.

As concrete values to demonstrate shortcomings of possible encodings, we will consider $l=0.0,\enspace u=1.0, \enspace c=0.5,$ and $r=0.5$ with $\varepsilon=0.1$, corresponding to 
the input region of the DNN from \Cref{fig:nn}.

\subsection{A Tale of Three Failed Encodings}
\label{subsec:algorithm:failed}
\looseness=-1
To provide a better intuition for our encoding, we first describe three encodings which are either semantically unsound or imprecise for zonotope propagation.

\paragraph{Common Generators.}
One may be tempted to initialize $\affineForm{}'=\affineForm{}''=\left(\left(r,0\right),c\right)$ with $\affineForm{}^\Delta=\left(\left(0,\varepsilon\right),0\right)$.
However, while $\left[\zonoLower{\affineForm{}^\Delta},\zonoUpper{\affineForm{}^\Delta}\right]$ would indeed be $\left[-\varepsilon,\varepsilon\right]$, this is unsound, as shown by the following example:

\begin{example}
Initializing the affine forms for our example yields $\affineForm{}'=\affineForm{}''=\left(\left(0.5,0\right),0.5\right)$ and $\affineForm{}^\Delta=\left(\left(0,0.1\right),0\right)$.
We can then compute a bound on the difference between inputs in the first NN $x$ and the second NN $y$ either via the bounds on $\affineForm{}^\Delta$ (yielding $\left[-0.1,0.1\right]$) or via bounds on the following difference:
\[\affineForm{}
'-\affineForm{}''=\left(\left(0.5 - 0.5, 0-0\right),0.5 - 0.5\right) = \left(\left(0,0\right),0\right).
\]
The latter produces bounds $\left[0,0\right]$.
Therefore, we would only propagate inputs with $\left|x-y\right|=0$ through the NN which unsoundly ignores other inputs with $\left|x-y\right|\leq 0.1$ (e.g.\ $x=0,\enspace y=0.1$).
\end{example}

\paragraph{Independent generators for $\affineForm{}'$ and $\affineForm{}''$.}
\looseness=-1
It is a logical consequence of this observation that $\affineForm{}',\affineForm{}''$ require independent generators.
A potential solution would be to use fully independent generator for both NNs (i.e.\ $\affineForm{}'=\left(\left(r,0,0\right),c\right)$ and $\affineForm{}''=\left(\left(0,r,0\right),c\right)$).
Two natural alternatives arise:
We can set $\affineForm{}^\Delta=\left(\left(r,-r,0\right),0\right)$ or we can set $\affineForm{}^\Delta=\left(\left(0,0,\varepsilon\right),0\right)$.
We will explain the shortcomings of both approaches using our example:
\begin{example}
Consider $\affineForm{}'=\left(\left(0.5,0,0\right),0.5\right)$, $\affineForm{}''=\left(\left(0,0.5,0\right),0.5\right)$ and $\affineForm{}^\Delta=\left(\left(0.5,-0.5,0\right),0\right)$:
In this case, $\affineForm{}^\Delta$ would yield the bounds $\left[-0.5-0.5,0.5+0.5\right] = \left[-1,1\right]$ which are much larger than the maximal difference $\left[-\varepsilon,\varepsilon\right]=\left[-0.1,0.1\right]$.
While we could reintroduce this information after propagation by appropriately constraining $\affineForm{}^\Delta$ (e.g. using linear programming), this would make the initial propagation of zonotopes through the NN extremely imprecise.

As an alternative, we could set $\affineForm{}^\Delta=\left(\left(0,0,\varepsilon\right),0\right)$ making $\affineForm{}^\Delta$ fully independent from the other affine forms.
However, the differential propagation algorithm uses expressions such as, e.g., $\affineForm{}^\Delta{} - \lambda \affineForm{}'$ which would result in the affine form $\left(\left(-\lambda 0.5, 0, 0.1\right),-\lambda 0.5\right)$.
Assuming $\lambda=0.5$, the bounds of this affine form would be
\[
\left[-\lambda 0.5-0.1-\lambda 0.5,0.1+\lambda 0.5-\lambda 0.5\right] = \left[-0.6,0.1\right].
\]
As will be shown below, these bounds are less precise than the bounds achievable with our encoding.
The lack of precision stems from the fact that the considered affine forms are fully independent.
\end{example}

\paragraph{An additional generator for $\affineForm{}''$.}
\looseness=-1
From our prior observations, it follows that neither fully dependent nor fully independent generators are desirable.
One may hence consider setting $\affineForm{}'=\left(\left(r,0\right),c\right)$, $\affineForm{}''=\left(\left(r,\varepsilon\right),c\right)$ and $\affineForm{}^\Delta = \left(\left(0,-\varepsilon\right),0\right)$.
While this approach is sound, it over-approximates the second DNN's input domain:

\begin{example}
Reconsidering our example, we would obtain the affine forms $\affineForm{}'=\left(\left(0.5,0\right),0.5\right)$, $\affineForm{}''=\left(\left(0.5,0.1\right),0.5\right)$ and $\affineForm{}^\Delta = \left(\left(0,-0.1\right),0\right)$.
In this case, $\affineForm{}''$ would yield the bounds
\[
\left[-0.5-0.1+0.5,0.5+0.1+0.5\right]=\left[-0.1,1.1\right] \supset \left[0,1\right]
\]
which is a strict superset of the range of interest $\left[l,u\right]$.
Hence, this approach, while sound, \emph{over-approximates} the input range of the second NN.
This can also lead to spurious counterexamples (e.g. if the NN is not robust for $x=0,\enspace y=-0.1<l$)
\end{example}
Notably, this shortcoming \emph{cannot} be mitigated by setting $\affineForm{}'=\left(\left(r-\varepsilon,0\right),c\right)$, $\affineForm{}''=\left(\left(r-\varepsilon,\varepsilon\right),c\right)$:
For our example this would yield the bounds $\left[0.1,0.9\right]\subset \left[0,1\right]$ for $\affineForm{}'$, which is too restrictive and hence unsound.

\paragraph{A Tale of Three Lessons.}
Above, we have seen that we must ensure that the generators of $\affineForm{}',\affineForm{}''$ and $\affineForm{}^\Delta$ are
\begin{enumerate*}
    \item sufficiently independent (to ensure soundness);
    \item sufficiently dependent (to ensure precise propagation of zonotopes);
    \item modeling the perturbation $\varepsilon$ \emph{symbolically} during propagation.
\end{enumerate*}
Based on these observations, we introduce \emph{differential halo zonotopes}, a specialized representation of input perturbations within the differential zonotopes abstract domain which soundly and precisely encodes perturbed input points.

\subsection{Differential Halo Zonotopes}
\label{subsec:algorith:diff_halo}
To achieve the three objectives outlined above simultaneously, we propose \emph{differential halo zonotopes}, which \emph{distribute} the perturbation $\varepsilon$ across two generators that are independent between $\affineForm{}'$ and $\affineForm{}''$, while they share a common primary generator.
Formally, we define a differential halo affine form as follows:
\begin{definition}[Differential Halo Affine Form]
\label{def:algorithm:diff_halo_zono}
For lower and upper bounds $l<u$ and a maximal perturbation $0 < \varepsilon < (u-l)$,
we set $c=\frac{u+l}{2}$ and $r=u-c$.
We then define a \emph{differential halo affine form} as $\left(\affineForm{}',\affineForm{}'',\affineForm{}^\Delta\right)$:
\begin{align*}
    \affineForm{}' &= \left(\left(\left(r-\frac{\varepsilon}{2}\right),\frac{\varepsilon}{2},0\right),c\right) &
    \affineForm{}'' &= \left(\left(\left(r-\frac{\varepsilon}{2}\right),0,\frac{\varepsilon}{2}\right),c\right)\\
    \affineForm{}^\Delta &= \left(\left(0,\frac{\varepsilon}{2},-\frac{\varepsilon}{2}\right),0\right) &
\end{align*}
\end{definition}
\noindent
We will call the first component of $\affineForm{}'$'s and $\affineForm{}''$'s generator vector the \emph{primary} generator while the latter two components are the \emph{secondary} generators.
Differential halo affine forms come with the following guarantee (proof in Isabelle~\cite{DBLP:books/sp/Paulson94}):
\begin{isaLemma}[Exact Encoding of Bounded Input Perturbations]
\label{lem:algorithm:diff_halo_zono}
Let $l,u,\varepsilon\in\mathbb{R}$ be bounds and perturbation such that $0 < \varepsilon < \left(u-l\right)$ and let $\left(\affineForm{}',\affineForm{}'',\affineForm{}^\Delta\right)$ be their differential halo affine form (see \Cref{def:algorithm:diff_halo_zono}), then for all $x,y\in\mathbb{R}$:
\begin{equation*}
    \left(x,y\right) \in \langle\left(\affineForm{}',\affineForm{}'',\affineForm{}^\Delta\right)\rangle
    \text{ iff }
    x,y \in \left[l,u\right] \text{ and } \left(x-y\right) \in \left[-\varepsilon,\varepsilon\right].
\end{equation*}
\end{isaLemma}

\begin{figure}[t]
\centering
\begin{tikzpicture}[x=0.4cm,y=0.4cm,
  box/.style={line width=0.9pt},
  axis/.style={->, line width=0.6pt},
  annot/.style={font=\small},
  zonoCenter/.style={fill=red!80!white,draw=none},
  zonoCommon/.style={fill=blue!20!white,draw=blue!80!white,line width=1.2pt},
  zonoCommonPoint/.style={fill=blue!80!white,draw=none},
]

\coordinate (L) at (0,0);
\coordinate (U) at (8,5);

\draw[axis] (-0.4,0) -- (8.7,0) node[anchor=west,annot] {$x_1$};
\draw[axis] (0,-0.4) -- (0,5.7) node[left,annot] {$x_2$};

\pic [zhalo/outline=false] at (0,0) {zhalo={w 8 h 5 inset 1}};

\draw[box,fill=none] (L) rectangle (U);

\fill[color=yellow!50!white, opacity=0.5]   (4.5,1.5) rectangle ++(2,2);
\fill[pattern={
Lines[angle=-45, distance=12pt,  line width=6pt]%
}, pattern color=orange!50!white, opacity=0.5
]  (4.5,1.5) rectangle ++(2,2);

\draw[zonoCommonPoint] (5.5,2.5) circle (0.1);

\node [annot,anchor=north east] at (0,0) {$\left(l_1,l_2\right)$};
\node [annot,anchor=north] at (8,0) {$u_1$};
\node [annot,anchor=east] at (0,5) {$u_2$};

\node [annot,anchor=north] at (1.5,0) {$(l_1 + \frac{\varepsilon}{2})$};
\node [annot,anchor=north] at (6.5,0) {$(u_1 - \frac{\varepsilon}{2})$};
\node [annot,anchor=east] at (0,1) {$(l_2 + \frac{\varepsilon}{2})$};
\node [annot,anchor=east] at (0,4) {$(u_2 - \frac{\varepsilon}{2})$};

\draw [-,dashed] (0,1) -- (8,1);
\draw [-,dashed] (0,4) -- (8,4);
\draw [-,dashed] (1,0) -- (1,5);
\draw [-,dashed] (7,0) -- (7,5);

\draw [<->,dashed] (4.5,3.5) -- node[anchor=south] {$\varepsilon$} (6.5,3.5);

\draw [<->,dashed] (4.5,3.5) -- node[anchor=east] {$\varepsilon$} (4.5,1.5);

\end{tikzpicture}
\caption{Intuition for \emph{Differential Halo Zonotopes}}
\label{fig:differential-halo-zonotope-intuition}
\end{figure}
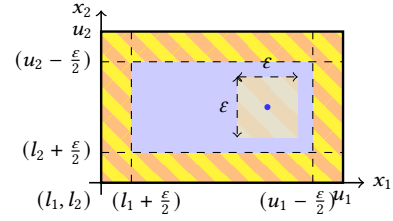

\begin{example}
We can observe the effect of differential halo zonotopes on our running example:
Instantiating the affine forms as described in \Cref{def:algorithm:diff_halo_zono} yields
$\affineForm{}'=\left(\left(0.45,0.05,0\right),0.5\right),\enspace
\affineForm{}''=\left(\left(0.45,0,0.05\right),0.5\right),
\affineForm{}^\Delta=\left(\left(0,0.05,0.05\right),0\right)$.
The bounds for $\affineForm{}'$ and $\affineForm{}''$ then are $\left[0,1\right]$ and the bounds of $\affineForm{}^\Delta$ are $\left[-0.1,0.1\right]$.
Additionally, we now preserve dependencies when adding affine forms:
\[
\affineForm{}^\Delta - \lambda \affineForm{}' = 
\left(
\left(
-\lambda 0.45,
(1-\lambda) 0.05,
0.05
\right),
-\lambda 0.5
\right).
\]
For $\lambda=0.5$ the affine form above yields bounds $\left[-0.55,0.05\right]$ which is tighter than $\left[-0.6,0.1\right]$ produced by the independent generator approach.
This difference in precision \emph{increases} with $\varepsilon$.
\end{example}

Multiple differential halo affine forms, can be independently combined to a differential halo zonotope describing a multi-dimensional (box) input space as usual.
\Cref{fig:differential-halo-zonotope-intuition} provides an intuition for the idea behind differential halo zonotopes for the two-dimensional case:
By fixing the primary generators of $\zonotope'$ and $\zonotope''$ to a fixed value, the blue dot in \Cref{fig:differential-halo-zonotope-intuition} can be reached.
We can then reach any point in the yellow-orange shaded area \emph{independently} with $\zonotope'$ and $\zonotope''$ by fixing the secondary generator.
This enables us to generate a perturbation of up to $\varepsilon$ (distributed across the secondary generators of $\zonotope'$ and $\zonotope''$).
At the same time, the perturbations ensure that the full input space between $l$ and $u$ is reachable, effectively producing a \emph{halo} of reachable points around each point reachable via the primary generators.

\begin{algorithm}
    \caption{Verification of Global Robustness via Differential Halo Zonotopes}
    \label{algo:algorithm}
    \begin{algorithmic}
    \Require Neural Network $f$, 
            Lower and Upper Bounds $\boldsymbol{l},\boldsymbol{u} \in \mathbb{R}^I$,\\
            Dimension-wise maximal perturbations $\boldsymbol{\varepsilon}\in\mathbb{R}^I$,\\
            Confidence Bounds $\tau_1,\tau_2\in\left[0.5,1\right)$\\
            ($\tau_2=\tau_1$ for symmetric property; $\tau_2=0.5$ for asymmetric property)
    \Procedure{$\textsc{Verify}_\Delta$}{$f,\boldsymbol{l},\boldsymbol{u},\boldsymbol{\varepsilon},\tau_1,\tau_2$}
        \State $Q \leftarrow \left\{\zonotope_{\text{in}}\right\} = \textsc{InitZono}\left(\boldsymbol{l},\boldsymbol{u},\boldsymbol{\varepsilon}\right)$
        \Comment{See \Cref{def:algorithm:diff_halo_zono} and \Cref{lem:algorithm:diff_halo_zono}}
        \While{$Q \ne \emptyset \land \neg\texttt{timeout}$}
            \State $\zonotope_{\text{cur}} \leftarrow \textsc{Pop}\left(Q\right)$
            \State $\zonotope_{\text{out}} \leftarrow \textsc{Reach}_\Delta\left(f,\zonotope{}_{\text{cur}}\right)$
            \Comment{See Algorithm 1 in \cite{teuber2025revisiting}}
            \State $r \leftarrow \textsc{CheckProp}\left(f,\zonotope_{\text{out}},\tau_1,\tau_2\right)$
            \Comment{See \Cref{lem:algorith:lp_soundness}}
            \If{r = \texttt{concrete cex}}
                \Return \texttt{UNSAFE}
            \EndIf{}
            \If{r = \texttt{unknown}}
                \State $\zonotope_1,\zonotope_2 \leftarrow \textsc{Refine}\left(\zonotope{}_{\text{cur}}\right)$
                \Comment{See \Cref{lem:algorithm:soundess_spits}}
                \State $Q \leftarrow Q \cup \left\{\zonotope_1,\zonotope_2\right\}$
            \EndIf{}
        \EndWhile
        \If{$Q = \emptyset$}
            \Return \texttt{SAFE}
        \Else{}
            \Return \texttt{UNKNOWN}
            \Comment{Timout reached}
        \EndIf{}
    \EndProcedure
    \end{algorithmic}
\end{algorithm}

\subsection{Verification Algorithm}
We provide an overview of our verification algorithm in \Cref{algo:algorithm}:
First, we initialize a differential halo zonotope $\zonotope_{\text{in}}$ which exactly represents the inputs in our bounding box $\left[l,u\right]\subset\mathbb{R}^I$ with perturbations in $\left[-\varepsilon,\varepsilon\right]$ via the insight from \Cref{lem:algorithm:diff_halo_zono} and add it to $Q$.
Subsequently, we run a loop which pops elements from $Q$, and propagates them through the DNN.
The propagation is visualized in \Cref{fig:algorithm:diffzono}: By propagating a bound on the \emph{difference} between executions in lock-step with the individual zonotopes, we obtain a bound on the difference between DNN outputs.
While zonotopes incur over-approximation errors through $\ReLU$ nodes (indicated by the dotted truly reachable region), the difference bound $\zonotope^\Delta$ is tighter than the bound obtained via $\zonotope' - \zonotope''$.
We then check the confidence-based robustness property on the resulting zonotope $\zonotope_{\text{out}}$ (see \Cref{subsubsec:algorith:property}).
This procedure either proves the absence of counterexamples for the zonotope, finds a concrete counterexample, or returns unknown.
In the first case, we can discard $\zonotope_{\text{cur}}$ as it was proven safe.
In the second case, we have a concrete violation of the property and abort the algorithm.
In the final case, we refine the zonotope to make the propagation more precise (see \Cref{subsubsec:algorith:refinement}).
We now describe the approach in more detail and prove soundness.
\definecolor{lightgray}{RGB}{180,180,180}

\begin{figure*}[t]
\centering
\resizebox{0.75\linewidth}{!}{
\begin{tikzpicture}[
    scale=0.95,
    zonotope/.style={draw,thick,fill=white},
    reachset/.style={fill=lightgray,draw=none},
    box/.style={fill=white,draw=black},
    diff/.style={draw,thick,fill=white},
    layer/.style={draw, rounded corners, minimum width=2.6cm, minimum height=0.8cm,align=center},
    arrow/.style={->, thick},
    netTransform/.style={align=center, draw, rounded corners, dashed},
    every node/.style={font=\small}
]

\node[anchor=east] at (-1.2,3.2) {$\zonotope'$};
\node[anchor=east] at (-1.2,1.6) {$\zonotope''$};
\node[anchor=east] at (-1.2,0) {$\zonotope^\Delta$};

\draw [box] (-0.8, 3.8) rectangle (0.8, 2.6);
\draw [box] (-0.8, 2.2) rectangle (0.8, 1.0);
\draw [box] (-0.5, 0.5) rectangle (0.5, -0.5);

\pic at (-0.4,2.9) {zhalo={w 0.8 h 0.6 inset 0.125}};

\pic at (-0.4,1.3) {zhalo={w 0.8 h 0.6 inset 0.125}};

\pic at (-0.25,-0.25) {haloonly={w 0.5 h 0.5}};

\node[below,align=center] at (1,-0.55) {Halo $\zonotope^\Delta$ encodes pertubation $\boldsymbol{\varepsilon}$};

\node (affineTransform1) [netTransform] at (3,3.2) {$\boldsymbol{x}' = W \boldsymbol{x} + \boldsymbol{b}$};
\node (affineTransform2) [netTransform] at (3,1.6) {$\boldsymbol{y}' = W \boldsymbol{y} + \boldsymbol{b}$};
\node (affineTransformd) [netTransform] at (3,0.0) {$\Delta = \boldsymbol{x}' - \boldsymbol{y}'$};

\draw [arrow] (0.8,3.2) -- (affineTransform1);
\draw [arrow] (0.8,1.6) -- (affineTransform2);
\draw [arrow] (0.5,0.0) -- (affineTransformd.west);
\draw [arrow] (0.8,1.6) -- (affineTransformd.west);

\draw [box] (5.2, 3.8) rectangle (6.8, 2.6);
\draw [box] (5.2, 2.2) rectangle (6.8, 1.0);
\draw [box] (5.5, 0.5) rectangle (6.5, -0.5);

\draw [arrow] (affineTransform1) -- (5.2,3.2);
\draw [arrow] (affineTransform2) -- (5.2,1.6);
\draw [arrow] (affineTransformd) -- (5.5,0.0);

\pic [anchor=center,yshift=-0.4cm,rotate=80] at (6,3.2) {zhalo={w 0.8 h 0.7 inset 0.125}};

\pic [anchor=center,yshift=-0.4cm,rotate=80] at (6.3,1.6) {zhalo={w 0.8 h 0.7 inset 0.125}};

\pic [rotate=20] at (5.9,-0.25) {haloonly={w 0.5 h 0.5}};

\node (reluTransform1) [netTransform] at (9,3.2) {$\boldsymbol{x}'' = \ReLU\left(\boldsymbol{x}'\right)$};
\node (reluTransform2) [netTransform] at (9,1.6) {$\boldsymbol{y}'' = \ReLU\left(\boldsymbol{y}'\right)$};
\node (reluTransformd) [netTransform] at (9,0.0) {$\Delta' = \boldsymbol{x}'' - \boldsymbol{y}''$};

\draw [arrow] (6.8,3.2) -- (reluTransform1);
\draw [arrow] (6.8,1.6) -- (reluTransform2);
\draw [arrow] (6.5,0.0) -- (reluTransformd.west);
\draw [arrow] (6.8,3.2) -- (reluTransformd.west);
\draw [arrow] (6.8,1.6) -- (reluTransformd.west);

\draw [box] (11.2, 3.8) rectangle (12.8, 2.6);
\draw [box] (11.2, 2.2) rectangle (12.8, 1.0);
\draw [box] (11.5, 0.5) rectangle (12.5, -0.5);

\draw [arrow] (reluTransform1) -- (11.2,3.2);
\draw [arrow] (reluTransform2) -- (11.2,1.6);
\draw [arrow] (reluTransformd) -- (11.5,0.0);

\pic (Z1) [anchor=center,yshift=-0.4cm,rotate=80,zhalo/outline=false] at (12,3.2) {zhalo={w 0.8 h 0.7 inset 0.125}};

\draw[densely dotted, thick]
  ($(Z1-L)$) -- ($(Z1-U |- Z1-L)+(-0.1,0.175)$) -- ($(Z1-U)$) -- cycle;

\pic (Z2) [anchor=center,yshift=-0.4cm,rotate=80,zhalo/outline=false] at (12.5,1.6) {zhalo={w 0.8 h 0.7 inset 0.125}};
\draw[densely dotted, thick]
  ($(Z2-L)$) -- ($(Z2-U |- Z2-L)+(-0.1,0.175)$) -- ($(Z2-U)$) -- ($(Z2-U)+(0.2,-0.05)$) -- cycle;

\pic (ZD) [rotate=20,haloonly/outline=false] at (11.7,-0.35) {haloonly={w 0.8 h 0.5}};

\draw[densely dotted, thick]
  ($(ZD-L)$) -- ($(ZD-U |- ZD-L)+(0.15,0.3)$) -- ($(ZD-U)$) -- ($(ZD-L)+(-0.05,0.15)$) -- cycle;

\end{tikzpicture}
}
\caption{Visualization of robustness verification with differential halo zonotopes:
Dotted border represents truly reachable points, halo zonotope (blue and yellow/orange) represents computed bounds. Bounding the difference in secondary generators via $\zonotope^\Delta$ yields tighter bounds on the difference between executions.}
\label{fig:algorithm:diffzono}
\end{figure*}
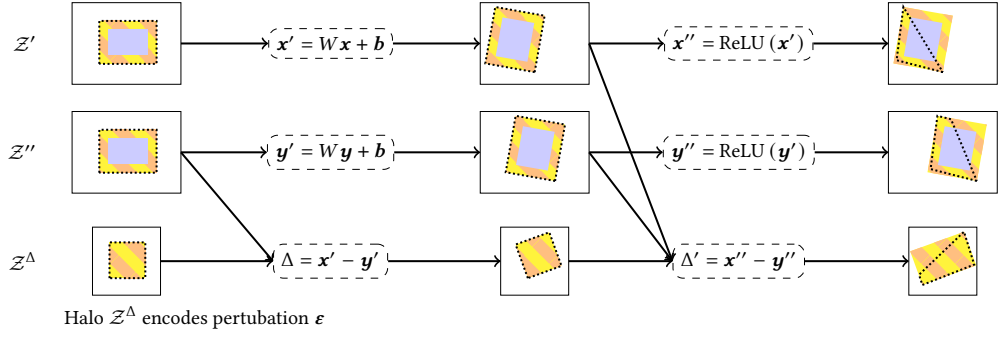

\subsubsection{Property Verification.}
\label{subsubsec:algorith:property}
To verify (symmetric) confidence-based global robustness, we encode the bounds obtained via $\textsc{Reach}_\Delta$ as an LP problem that proves the absence of adversarial attacks with sufficient confidence.
Using the $\mathrm{softmax}$-approximation by Teuber \emph{et al.}~\cite{teuber2025revisiting}, we encode the absence of attacks where $j$ is maximal in the first NN and $k$ is maximal in the second NN as follows:

\begin{definition}[Symmetric Confidence LP Formulation]
\label{def:algorithm:confidence_lp}
Assume $\tau_1,\tau_2 \in \left[0.5,1\right)$ and let 
\[
\zonotope_{\text{out}} = \left(\zonotope_{\text{out}}',\zonotope_{\text{out}}'',\zonotope_{\text{out}}^\Delta\right)
=\left(\left(G',\boldsymbol{c}'\right),\left(G'',\boldsymbol{c}''\right),\left(G^\Delta,\boldsymbol{c}^\Delta\right)\right)
\]
be some differential zonotope (with aligned matrices).
For $j,k\in\left[1,O\right]$ with $j \neq k$ we define the \emph{Symmetric Confidence LP} as follows:
\begin{align}
    G'_l \boldsymbol{x} + \boldsymbol{c}'_l + \ln\left({\tau_1}/{\left(1-\tau_1\right)}\right) &\leq 
    G'_j \boldsymbol{x} + \boldsymbol{c}'_j & \text{for }l \neq j
    \label{eq:algorith:confidence_lp:net1}\\
    G''_l \boldsymbol{x} + \boldsymbol{c}''_l + \ln\left({\tau_2}/{\left(1-\tau_2\right)}\right) &\leq
    G''_k \boldsymbol{x} + \boldsymbol{c}''_k & \text{for } l \neq k
    \label{eq:algorith:confidence_lp:net2}\\
    G^\Delta\boldsymbol{x} + \boldsymbol{c}^\Delta &= \rlap{
    $\left(G'\boldsymbol{x} + \boldsymbol{c}'\right) - \left(G''\boldsymbol{x} + \boldsymbol{c}''\right)$}
    \label{eq:algorith:confidence_lp:diff}
\end{align}
\end{definition}
\noindent
The constraints (\ref{eq:algorith:confidence_lp:net1}) and (\ref{eq:algorith:confidence_lp:net2}) assert a distance larger than $\ln\left({\tau_i}/{\left(1-\tau_i\right)}\right)$ between the maximal output ($j$ or $k$) and the second largest output, which over-approximates high-confidence outputs.
Constraint (\ref{eq:algorith:confidence_lp:diff}) bounds the difference between the two NN outputs with $\zonotope^\Delta$.
We can then verify global robustness with this LP (for asymmetric robustness choose $\tau_2=0.5$; see proof on \cpageref{proof:algorith:lp_soundness}):
\begin{lemmaE}[Soundness of Symmetric Confidence LP][end,restate,text link={}]
\label{lem:algorith:lp_soundness}
Assume $\tau_1,\tau_2 \in \left[0.5,1\right)$ and let $\zonotope_{\text{out}}$ be a differential zonotope returned for $\textsc{Reach}_\Delta\left(f,\zonotope_{\text{in}}\right)$.
If for all $j,k\in\left[1,O\right]$ with $j \neq k$ the symmetric confidence LP is infeasible, then for all $\left(\boldsymbol{x},\boldsymbol{y}\right) \in \langle \zonotope_{\text{in}} \rangle$ (where $\tau_i^*$ is $0$ if $\tau_i=0.5$ and otherwise $\tau_i^*=\tau_i$ for $i \in \left\{1,2\right\}$):
\begin{align*}
    \mathrm{conf}\left(f\left(\boldsymbol{x}\right)\right) > \tau_1^* \land
    \mathrm{conf}\left(f\left(\boldsymbol{y}\right)\right) > \tau_2^* &
    \rightarrow 
    \\ \mathrm{class}\left(f\left(\boldsymbol{x}\right)\right) = 
    \mathrm{class}\left(f\left(\boldsymbol{y}\right)\right),
\end{align*}
\end{lemmaE}
\begin{proofE}
\label{proof:algorith:lp_soundness}
\sloppy
Note that via \cite[Theorem 2 and Corollary 3]{teuber2025revisiting}, $\zonotope_{\text{out}}$ over-approximates the behavior of $f$ for the two inputs $\left(\boldsymbol{x},\boldsymbol{y}\right)\in\langle\zonotope_{\text{in}}\rangle$, i.e. for $\boldsymbol{v}\in\mathbb{R}^n$ we know that
$f\left(\zonotope_{\text{in}}'\left(\boldsymbol{v}\right)\right) \in
\langle\zonotope_{\text{out}}'\left(\boldsymbol{v}\right)\rangle$,
$f\left(\zonotope_{\text{in}}''\left(\boldsymbol{v}\right)\right) \in
\langle\zonotope_{\text{out}}''\left(\boldsymbol{v}\right)\rangle$, and finally
$\left(f\left(\zonotope_{\text{in}}'\left(\boldsymbol{v}\right)\right) - 
f\left(\zonotope_{\text{in}}''\left(\boldsymbol{v}\right)\right)\right) \in
\langle\zonotope_{\text{out}}^\Delta\left(\boldsymbol{v}\right)\rangle$ with the \emph{same} generator values $\boldsymbol{\epsilon}$.
Hence, the points reachable for $f\left(\zonotope_{\text{in}}''\left(\boldsymbol{v}\right)\right)$ can also be described by the zonotope $\tilde{\zonotope}_{\text{out}}'' = \left(G'-G^\Delta,\boldsymbol{c}' - \boldsymbol{c}^\Delta\right)$ (assuming aligned matrices).
Thus, any reachable values of the two executions of $f$ satisfy the constraint ``$\zonotope^\Delta = \zonotope' - \zonotope''$'', formally $G^\Delta\boldsymbol{x} + \boldsymbol{c}^\Delta = 
    \left(G'\boldsymbol{x} + \boldsymbol{c}'\right) - \left(G''\boldsymbol{x} + \boldsymbol{c}''\right)$.
The constraint $G'_l \boldsymbol{x} + \boldsymbol{c}'_l + \ln\left({\tau_1}/{\left(1-\tau_1\right)}\right) \leq 
G'_j \boldsymbol{x} + \boldsymbol{c}'_j$ (for $l \neq j$) then over-approximates the region where $f\left(\zonotope'_{\text{in}}\left(\boldsymbol{v}\right)\right)$ has maximum $j$ with a gap of at least $\ln\left({\tau_1}/{\left(1-\tau_1\right)}\right)$.
As per \cite[Lemma 4]{teuber2025revisiting}, this over-approximates the region for which $\mathrm{softmax}_j\left(G'\boldsymbol{x} + \boldsymbol{c}'\right) \geq \tau_1$.
With the same argument, 
the constraint
$G''_l \boldsymbol{x} + \boldsymbol{c}''_l + \ln\left({\tau_2}/{\left(1-\tau_2\right)}\right) \leq 
G''_k \boldsymbol{x} + \boldsymbol{c}''_k$
then over-approximates the region for which
$\mathrm{softmax}_j\left(G''\boldsymbol{x} + \boldsymbol{c}''\right) \geq \tau_2$.
Hence, for $k \ne j$, the infeasibility of the LP from \Cref{def:algorithm:confidence_lp} implies there exist no $\boldsymbol{v}$ such that $\mathrm{conf}\left(f\left(\zonotope_{\text{in}}'\left(\boldsymbol{v}\right)\right)\right) > \tau_1$, $\mathrm{conf}\left(f\left(\zonotope_{\text{in}}''\left(\boldsymbol{v}\right)\right)\right) > \tau_2$, $\mathrm{class}\left(f\left(\zonotope_{\text{in}}'\left(\boldsymbol{v}\right)\right)\right) = j$, and $\mathrm{class}\left(f\left(\zonotope_{\text{in}}''\left(\boldsymbol{v}\right)\right)\right) = k$.
As an immediate consequence, there are no inputs $\left(\boldsymbol{x},\boldsymbol{y}\right) \in \langle\zonotope_{\text{in}}\rangle$ such that $\mathrm{conf}\left(f\left(\boldsymbol{x}\right)\right) > \tau_1$, $\mathrm{conf}\left(f\left(\boldsymbol{y}\right)\right) > \tau_2$, $\mathrm{class}\left(f\left(\boldsymbol{x}\right)\right) = j$, and $\mathrm{class}\left(f\left(\boldsymbol{y}\right)\right) = k$.
Hence, if the LP infeasibility can be shown for all $j,k\in\left[1,O\right]$ with $j\neq k$, then we get the implication of \Cref{lem:algorith:lp_soundness}.
Regarding the special case of $\tau_i=0.5$, note that $\ln\left(0.5/\left(1-0.5\right)\right)=\ln\left(1\right)=0$.
Hence, for $\tau_i=0.5$ we impose no constraint on the output of the first/second execution beyond maximality of output dimension $j$/$k$.
Hence, in this case the class invariance across the two executions holds for any confidence threshold.
\end{proofE}
The \textsc{CheckProp} from \Cref{algo:algorithm} then consists of constructing the  $\mathcal{O}\left(O^2\right)$ LP problems and checking for infeasibility.
To speed up verification, we cache already proven cases.
Additionally, we first construct an LP with only \Cref{eq:algorith:confidence_lp:net1,eq:algorith:confidence_lp:diff} and optimize the violation of \Cref{eq:algorith:confidence_lp:net2} for dimension $j$ (which is maximal in the first DNN).
Then, we only construct the full LP if dimension $j$ is feasible.
When solving the full LP, we optimize output dimension $k$ of the second DNN to increase the likelihood of finding counterexamples.
If a concrete counterexample (cross-checked via execution of $f$) is found, \textsc{CheckProp} returns \texttt{concrete cex}, if all LPs are infeasible it returns \texttt{safe}, otherwise it returns \texttt{unknown}.

\subsubsection{Zonotope Refinement.}
\label{subsubsec:algorith:refinement}
While the over-approximations of zonotopes ensure soundness, i.e.\ whenever we prove a DNN to be robust it is indeed robust, they can compromise completeness, i.e.\ there may be instances where we cannot verify the robustness of a DNN that is indeed robust.
To combat this challenge, many approaches employ \emph{input splitting} which partitions the input space into two regions that are approximated and verified independently.

While this approach can reduce over-approximation error and hence improve the share of verifiable DNNs, special care is required when refining global robustness queries.
Since robustness properties consider pairs of inputs that differ by up to $\varepsilon$, naively splitting input spaces of both DNNs simultaneously can be unsound as this omits adversarial examples that straddle the splitting boundary.
Instead of splitting the input space, we propose to split the differential halo zonotope directly to enable sound refinement by splitting \emph{primary and secondary generators} instead of splitting {input dimensions}:
\begin{definition}[Generator Split]
\label{def:algorithm:diff_halo_generator_split}
Let $\left(\affineForm{}',\affineForm{}'',\affineForm{}^\Delta\right)$
be a differential affine form, where each affine form is of the shape $\affineForm{} = \left((g_1,g_2,g_3),c\right)$.
The \emph{generator split} of the differential affine form along generator $j \in \left[1,3\right]$ yields two
differential halo affine forms
$
\affineForm{}^- = \left(\affineForm{}'^{-},\affineForm{}''^{-},\affineForm{}^{\Delta -}\right)$ and $
\affineForm{}^+ = \left(\affineForm{}'^{+},\affineForm{}''^{+},\affineForm{}^{\Delta +}\right)
$
defined by replacing the $j$-th generator $g_j$ with $\hat{g}_j$ and center $c$ with $\hat{c}$ in each component
$\affineForm{}', \affineForm{}'', \affineForm{}^\Delta$ as follows while leaving all other generators unchanged:
\begin{align*}
\hat{g}_j &\coloneqq
\;\frac{g_j}{2}&
\hat{c} & \coloneqq
\begin{cases}
c - \frac{g_j}{2} & \text{in } \affineForm{}^-\\[0.3em]
c + \frac{g_j}{2} & \text{in } \affineForm{}^+
\end{cases}
\end{align*}
\end{definition}
\noindent
Each tuple of inputs $(x,y)$ contained in the original zonotope is then by construction contained in (at least) one of the resulting zonotopes (proof in Isabelle~\cite{DBLP:books/sp/Paulson94}):
\begin{isaLemma}[Soundness of Generator Splits]
\label{lem:algorithm:soundess_spits}
\looseness=-1
Let $\affineForm{} = \left(\affineForm{}',\affineForm{}'',\affineForm{}^\Delta\right)$ be a differential affine form.
If we split $\affineForm{}$ along $j \in \left[1,3\right]$ yielding
$
\affineForm{}^- = \left(\affineForm{}'^{-},\affineForm{}''^{-},\affineForm{}^{\Delta -}\right),
\affineForm{}^+ = \left(\affineForm{}'^{+},\affineForm{}''^{+},\affineForm{}^{\Delta +}\right)
$,
then for any $\left(x,y\right) \in \langle\affineForm{}\rangle$ we get
$\left(x,y\right) \in \langle\affineForm{}^{-}\rangle$ or $\left(x,y\right) \in \langle\affineForm{}^{+}\rangle$.
\end{isaLemma}
\noindent
In contrast to classical input splitting, splitting the generators of a differential halo zonotope creates a small (and in the number of splits decreasing) amount of overlap between the resulting input spaces.
This overlap is necessary to preserve coverage of all possible $\varepsilon$-perturbations w.r.t. all possible inputs.

\paragraph{Input Splitting Heuristic}
\looseness=-1
We adapt the input splitting heuristic from Teuber \emph{et al.}~\cite{teuber2025revisiting}:
Previously, the heuristic computed the \emph{influence} of each input dimension by initializing an $I$-dimensional unit matrix representing the influence of each generator (column) on each input dimension (row).
Subsequently, new columns were appended for each newly introduced generator.
We modify the heuristic to compute the influence of initial primary and secondary generators by introducing additional rows and columns for secondary generators in the initial influence matrix.
Because primary and secondary generators typically operate on different magnitudes, we normalize the influence matrix by their average scale.
We also extend the heuristic computation rule of \cite[Apx. D]{teuber2025revisiting} with an additional normalization by the zonotope's output range to improve numerical stability (previously, the heuristic typically reached values $> 10^7$).
The primary or secondary generator to split is then chosen by the heuristic's maximal influence value.

\subsubsection{Soundness.}
Using the results established above together with the results from  \Cref{subsec:algorith:diff_halo},
we get the following soundness result for \Cref{algo:algorithm} (note that the algorithm may also return \texttt{UNKNOWN} in case of a timeout):
\begin{theoremE}[Soundness of $\textsc{Verify}_\Delta$][end,restate,text link={}]
\label{thm:algorithm:overall_soundness}
If $\textsc{Verify}_\Delta$ returns \texttt{SAFE} for input range $\left[\boldsymbol{l},\boldsymbol{u}\right]$ and perturbation $\boldsymbol{\varepsilon}$ then $f$ satisfies:
\begin{enumerate*}
    \item symmetric confidence-based robustness for confidence $\tau_1$ if $\tau_2=\tau_1$;
    \item asymmetric confidence-based robustness for confidence $\tau_1$ if $\tau_2 = 0.5$.
\end{enumerate*}
\texttt{UNSAFE} implies it found a concrete counterexample.
\end{theoremE}
\begin{proof}[Proof sketch]
The last statement (correctness of \texttt{UNSAFE}) follows directly from the observation that \textsc{CheckProp} returns \texttt{concre\-te cex} iff it found a concrete robustness violation.
For the first two statements (\texttt{SAFE}), we construct $R$ as the set of all zonotopes that were removed without replacement.
Let $\langle R \rangle = \left\{ \left(x,y\right) \in \langle \zonotope\rangle \middle| \zonotope\in R\right\}$.
We then have two invariants:
First, all points in $\langle R \rangle$ have been proven robust.
Secondly, $\langle \zonotope_{\text{in}}\rangle \subseteq \langle Q\rangle \cup \langle R\rangle$ is a loop invariant.
It follows that we have proven robustness if $Q = \emptyset$.
For a full proof see \cpageref{proof:algorithm:overall_soundness}.
\end{proof}
\begin{proofEnd}
\label{proof:algorithm:overall_soundness}
As per \Cref{lem:algorithm:diff_halo_zono}, we know that $\langle\zonotope_{\text{in}}\rangle$ exactly contains all $\boldsymbol{x},\boldsymbol{y}\in\left[l,u\right]$ with $\left(\boldsymbol{x}-\boldsymbol{y}\right)\in\left[-\varepsilon,\varepsilon\right]$.
We will first address the case where \texttt{SAFE} is returned, i.e. we can assume that \textsc{CheckProp} never returns \texttt{concrete cex}.
We identify the queue $Q$ after the $i$-th iteration as $Q_i$ (with $Q_1 = Q$).
From \Cref{lem:algorith:lp_soundness} we then know that an element $\zonotope_{\text{cur}}$ is removed from $Q$ without replacement iff confidence-based global robustness w.r.t. $\tau_1,\tau_2$ has been established for all $\left(\boldsymbol{x},\boldsymbol{y}\right) \in \langle \zonotope_{\text{cur}} \rangle$.
Let $R_i$ be the set of all $\zonotope_{\text{cur}}$ removed in this way without replacement after the $i$-th iteration (with $R_1=\emptyset$).
Then define $\langle R_i \rangle \coloneqq \left\{ \left(\boldsymbol{x},\boldsymbol{y}\right)\in\langle\zonotope\rangle \middle| \zonotope \in R_i\right\}$ (accordingly for $Q$ and $Q_i$).
After any iteration $i$ we know that we have established confidence-based robustness w.r.t. $\tau_1,\tau_2$ for all $\left(\boldsymbol{x},\boldsymbol{y}\right) \in \langle R_{i+1} \rangle$.
Further, in case of a split, we know that $\langle\zonotope_{\text{cur}}\rangle \subseteq \langle\zonotope_1\rangle\cup\langle\zonotope_2\rangle$ from \Cref{lem:algorithm:soundess_spits}.
Hence, we always have $\langle Q_i \rangle \subseteq \langle Q_{i+1} \rangle \cup \langle R_{i+1} \rangle$ (in practice, this is even an equality, however this is not needed for soundness).
The algorithm returns \texttt{SAFE} only if $Q = \emptyset$, but then $\langle Q \rangle = \emptyset$ and hence $\langle Q \rangle \subseteq \langle R_i \rangle$ for some $R_i$.
But since $\langle Q \rangle = \langle \zonotope_{\text{in}} \rangle$ contains all $\boldsymbol{x},\boldsymbol{y}$ within bounds and with pertubations bounded by $\varepsilon$, we have proven confidence-based robustness w.r.t. $\tau_1,\tau_2$ for all such $\boldsymbol{x},\boldsymbol{y}$.

Regarding the case where \textsc{PropCheck} returns \texttt{concrete cex}, the procedure always checks whether the returned counterexample is a concrete violation w.r.t. $f$.
Hence, any such returned values represent a concrete counterexample.
\end{proofEnd}

\section{Evaluation}
\label{sec:evaluation}
\looseness=-1
We implement our verification algorithm from \Cref{sec:algorithm} on top of the propagation algorithm from \verydiff~\cite{teuber2025revisiting} in Julia~\cite{bezanson2017julia} using Gurobi~\cite{gurobi} as LP backend to enable the verification of confidence-based robustness properties.
We will refer to our implementation as \ourTool{}.
We evaluate the performance of \ourTool{} on a diverse set of publicly available DNNs. We evaluate on both the asymmetric and symmetric  %
confidence-based global robustness properties as defined in \Cref{def:conf_rob_old,def:conf_rob_new}, respectively. We also compare our method to the seminal work by Athavale et al. \cite{athavale2024verifying}, which introduced the concept of confidence-based global robustness and the first verification technique thereof. %
We refer to their approach as the \emph{CGR (confidence-based global robustness) baseline} throughout the evaluation.
We also extended \emph{CGR baseline} to support symmetric confidence-based global robustness.

The experiments were conducted on an AMD EPYC 7713 64-Core Processor with 32 GB RAM. The timeout was set to 2 hours. 
To compare to the CGR baseline, we convert DNNs to nnet format that their tool requires.
Overall, we observe that the CGR baseline often generates spurious counterexamples due to the coarser approximation in their implementation.
For all cases where \ourTool{} proves safety and the CGR baseline reports a counterexample, a manual inspection found the counterexamples generated by the baseline to be spurious. We marked these cases by $^\star$.

\subsection{Datasets and Neural Networks}
We evaluate \ourTool{} on benchmarks drawn from prior literature. 
While numerous benchmarks exist for evaluating \emph{local} robustness of DNNs, there is currently no standardized benchmark set for evaluating \emph{global} robustness properties. Hence, we select a set of publicly available benchmark DNNs used in prior work on global or relational DNN verification, to
cover different regimes of DNN size, input dimensionality, and output
structure. Below, we describe the benchmarks in detail. In contrast to local robustness verification, for each DNN we fix a large input space box (e.g., the entire normalized input space $[-1,1]^{561}$ for HAR) for which we \emph{globally} prove the considered robustness property (i.e., expressed as a double  quantification over all possible neighboring  pairs in the input space, as opposed to a single quantification over the region around a specific input, as in the case of local robustness).

\paragraph{Law.}
\looseness=-1
We use the Law benchmark from the baseline tool~\cite{athavale2024verifying}.
The Law School Admissions Council (LSAC) provides a dataset called Law School Admissions comprising information on law students.
The dataset tracks the students’ progress through law school and predicts their likelihood of passing the bar exam~\cite{wightman1998lsac}.
The DNN is tiny with 4 inputs, 4-3 hidden neurons and 2 outputs.
The normalized input region is  $[-1,1]^4$.
We include this DNN despite its size for comparison because the CGR baseline was able to verify it.

\paragraph{LHC}
Since the Law benchmark is relatively small, we evaluate our approach on a larger DNN to assess scalability.
To this end, we investigate the CERN LHC benchmark~\cite{duarte2018fast} previously used for DNN verification~\cite{teuber2025revisiting}. 
The LHC benchmark consists of a DNN designed for particle classification based on high-dimensional detector features.
It is optimized for deployment on hardware accelerators (e.g., FPGAs) to meet strict timing constraints. Such resource-constrained deployment environments necessitate 
compact network architectures by design. Yet, the correctness-critical nature of particle physics experiments demands strong guarantees on these networks -- making them a natural target for 
global robustness verification. The trained DNN has 16 input features, two hidden layers with 40 $\ReLU$s each,
and 5 output classes. The normalized input domain is $[-0.5,0.5]^{16}$.

\paragraph{ACAS Xu.} To further evaluate our scalability, we pick VNN-COMP's~\cite{kaulen20256th} most complex fully-connected DNN, namely, ACAS from Julian et al.~\cite{julian2016policy}:
This is a standard benchmark widely used by many DNN verification tools.
The DNNs are trained for airborne collision avoidance, where they take aircraft state information as input, and output advisories corresponding to collision-avoidance maneuvers.
The DNN has 5 inputs, 6 hidden layers with 50 neurons each, and 5 outputs.
We use this benchmark to investigate how we perform on deeper DNNs.
All inputs are scaled to the $[-1,1]$ range using the standard ACAS Xu normalization prior to verification. Specifically, the normalized input domain of interest is given by:
$X_0 \in [0.6,0.6799],\ 
X_1,X_2 \in [-0.5,0.5],\ 
X_3 \in [0.45,0.5],\ 
X_4 \in [-0.5,-0.45]$.

\paragraph{HAR} Finally we consider the HAR benchmark~\cite{DBLP:conf/esann/AnguitaGOPR13} with its large input space.
HAR is a publicly available benchmark that has been used in prior DNN verification work \cite{paulsen_reludiff_2020,kern2025certified,paulsen_neurodiff_2020}.
The Human Activity Recognition (HAR) dataset tries to predict the activity of a given individual (walking, standing, etc.) from 561 sensory inputs. As HAR systems are increasingly deployed in healthcare and assisted living applications~\cite{de2017mobile, serpush2022wearable, liu2021overview, bibbo2023human,sankaran2025future}, where incorrect activity predictions can have serious consequences making strong robustness guarantees particularly valuable. At the same time, deployment on wearable and embedded sensing devices imposes strict memory and power constraints, motivating the use of lightweight architectures, making this benchmark a challenging and practically relevant target for global robustness verification. %
The considered DNN was trained by Kern \emph{et al.}~\cite{kern2025certified} and has 500 $\ReLU$ nodes.
The normalized input region is  $[-1,1]^{561}$.
In contrast to the previously mentioned DNNs, this DNN was trained with L1-regularization 
which has been shown to be beneficial for other DNN verification problems~\cite{kern2025certified}.

\subsection{Asymmetric Confidence-Based Global Robustness} 
We evaluate \ourTool{} for various combinations of benchmarks, epsilon values $\varepsilon$, and confidence values $\tau$ on the \textit{asymmetric confidence-based robustness} property. %
Our results can be found in \Cref{tab:evaluation:all_experiments}. %
We first evaluate \ourTool{} on the Law benchmark~\cite{athavale2024verifying}.
This is a small DNN with 4 inputs, 4-3 hidden neurons, and 2 outputs.
As shown in \Cref{tab:evaluation:all_experiments},
in this setting, both tools successfully verify the asymmetric property. However, \ourTool{} achieves lower verification times, reducing runtime by
approximately 30–40\% compared to the baseline.
Since Athavale \emph{et al.}~\cite{athavale2024verifying} only evaluated on DNNs with up to 15 $\ReLU$ nodes, this raises the question of whether \ourTool{} can scale to large DNNs.

To evaluate this question, we ran both tools on the
LHC benchmark~\cite{duarte2018fast,teuber2025revisiting}, a DNN with 16 inputs, 5 outputs, and 2 hidden layers with 40 $\ReLU$ nodes each.
As shown in \Cref{tab:evaluation:all_experiments},
\ourTool{} proves asymmetric confidence-based global robustness for the parameter pairs $(\varepsilon, \tau) = \left(0.001, 0.9\right)$ and $\left(0.001, 0.99\right)$ in less than 3 seconds. 
In contrast, the baseline tool times out after 2 hours for $\left(0.001, 0.99\right)$ and
otherwise reports spurious counterexamples due to the coarser $\mathrm{softmax}$ approximation
(since the approximation error of the \textit{CGR baseline} grows linearly with the number of output classes and reaches up to 0.40 for DNNs with 5 outputs, leading to spurious violations).

To scale the approach even further, we evaluate on the ACAS benchmark, a DNN with 6 hidden layers and 300 $\ReLU$ nodes.
We see in \Cref{tab:evaluation:all_experiments} that \ourTool{} is able to prove (asymmetric) robustness for confidence values $\tau \geq 0.7$ with epsilon 0.001 in 6742 sec and for confidence $\tau \geq 0.99$ in 182 sec.
In contrast, CGR baseline times out in both cases.

To further investigate scalability in the DNN's input dimensionality,
we run both tools on the HAR benchmark, a DNN with 561 inputs, 6 output classes, and 500 $\ReLU$ nodes.
As shown in \Cref{tab:evaluation:all_experiments}, \ourTool{} is capable of proving (asymmetric) robustness for $\varepsilon = 0.001$ and $\tau \geq 0.6$, while the CGR baseline exhausts memory in all cases. For $(\varepsilon, \tau) = \left(0.01, 0.99\right)$, \ourTool{} proves (asymmetric) robustness in 2.50 seconds. The baseline tool, however, runs out of memory.

\begin{table*}[t]
    \caption{Verification results for asymmetric and symmetric confidence-based robustness with a 2h timeout (\ourTool{} is our approach, Baseline is the \emph{CGR baseline}~\cite{athavale2024verifying}).}
    \centering
    
    \resizebox{0.8\textwidth}{!}{%
    \footnotesize    
    \begin{tabular}{c|c c||ll ll||ll ll}
        \multirow{2}{*}{\textbf{DNN}} &
        \multirow{2}{*}{$\varepsilon$} & \multirow{2}{*}{$\tau$} &
        \multicolumn{4}{c||}{\textbf{Asymmetric Robustness}} & 
        \multicolumn{4}{c}{\textbf{Symmetric Robustness}}\\\cline{4-7}\cline{8-11}
        &&&
        \multicolumn{2}{c}{\textbf{\ourTool{}}} &
        \multicolumn{2}{c||}{\textbf{Baseline}} &
        \multicolumn{2}{c}{\textbf{\ourTool{}}} &
        \multicolumn{2}{c}{\textbf{Baseline}}\\\hline\hline

         \multirow{2}{*}{\rotatebox[origin=c]{90}{{Law}}} &
        $10^{-3}$ &
        0.6 &
        \textbf{Safe} & ({1.34s}) &
        Safe & (2.21s) &
        Safe & (2.02s) & \textbf{Safe} & {(1.97s) }  \\
     
        &
        $0.1$ &
        0.6 &
        \textbf{Safe} & (1.45s) &
        Safe & (2.00s) &
        \textbf{Safe} & (1.99s) & Safe & (2.32s) \\\hline

        \multirow{5}{*}{\rotatebox[origin=c]{90}{{LHC}}} &
         \multirow{5}{*}{$10^{-3}$}& 
         0.6 &
         TO & & Cex$^\star$ & (1.54s) &
         TO & & Cex$^\star$ & (2.72s)\\
         &&
         0.7 &
         TO & & Cex$^\star$ & (1.32s) &
         TO & & Cex$^\star$ & (2.68s)\\
         &&
         0.8 &
         TO & & Cex$^\star$ & (1.16s) &
         \textbf{Safe} & (58.35s) & Cex$^\star$ & (2.67s)\\
         &&
         0.9 &
         \textbf{Safe} & (0.70s) & Cex$^\star$ & (1.30s) &
         \textbf{Safe} & (37.90s) & Cex$^\star$ &   (2.67)      \\
         &&
         0.99 &
         \textbf{Safe} & (2.35s) & TO & &
          \textbf{Safe} & (5.40s) & TO 
          \\\hline
        
        \multirow{3}{*}{\rotatebox[origin=c]{90}{{ACAS}}} &
        \multirow{3}{*}{$10^{-3}$} &
        0.6 &
        TO & &
        TO & &
        \textbf{Safe} & (6914s) & 
        Cex$^\star$ & (5.11s) \\
        &&
        0.7 &
        \textbf{Safe} & (6742s) &
        TO & &
        \textbf{Safe} & (6355s) &
        Cex$^\star$ & (5.03s) \\
        &&
        0.99 & 
        \textbf{Safe} & (182s) &
        TO & &
        \textbf{Safe} & (146s) &
        TO\\\hline

        \multirow{10}{*}{\rotatebox[origin=c]{90}{{HAR}}} &
        \multirow{5}{*}{$10^{-3}$} &
        0.6 &
        \textbf{Safe}   & (1.40s)   & MO   &    &
        \textbf{Safe}   & (11.05s)  & MO  &   \\
        &&
        0.7 &
        \textbf{Safe}   & (1.39s)   & MO   &    &
        \textbf{Safe}   & (12.44s)  & MO   &    \\
        &&
        0.8 &
        \textbf{Safe}   & (1.43s)   & MO   &    &
        \textbf{Safe}   & (10.48s)  & MO   &   \\
        &&
        0.9 &
        \textbf{Safe}   & (1.40s)   & MO  &  &
        \textbf{Safe}   & (9.847s)  & MO   &   \\
        &&
        0.99 &
        \textbf{Safe}   & (2.36s)   & MO    &    &
        \textbf{Safe}   & (6.25s)   & MO 
        \\\cline{2-11}
        &
        \multirow{5}{*}{0.01} &
        0.6 &
        MO              &           & MO    &    &
        \textbf{Cex}             & (11.32s)  & MO           & \\
        &&
        0.7 &
        MO              &           & MO    &   &
        TO              &           & MO   & \\
        &&
        0.8 &
        TO              &           & MO    &    &
        TO              &           & MO    &      \\
        &&
        0.9 &
        TO              &           & MO    &      &
        \textbf{Safe}   & (10.75s)  & MO    &      \\
        &&
        0.99 &
        \textbf{Safe}   & (2.50s)   & MO     &      &
        \textbf{Safe}   & (7.24s)   & MO
    \end{tabular}
    }
    \label{tab:evaluation:all_experiments}
\vspace{0.5ex}
 \begin{minipage}{0.85\textwidth}
    \footnotesize $^\star$Spurious counterexamples
\end{minipage}
\end{table*}

\subsection{Symmetric Confidence-Based Robustness}
We also evaluate \ourTool{} for symmetric confidence-based robustness. Our results empirically support the motivation for introducing the symmetric formulation (see \Cref{tab:evaluation:all_experiments}).

On the HAR benchmark, the asymmetric property is inconclusive for both $(\varepsilon, \tau) = (0.01, 0.9)$ and 
$(\varepsilon, \tau) = (0.01, 0.6)$. In contrast, \ourTool{} verifies the symmetric property in the former case in approximately 11 seconds, and finds a counterexample in the latter case in approximately 11 seconds. The CGR baseline exhausts memory in all cases. 

We observe similar results on the ACAS and LHC benchmarks: 
\ourTool{} verifies the symmetric property for $(\varepsilon, \tau) = (0.001, 0.6)$ on ACAS and for 
$(\varepsilon, \tau) = (0.001, 0.8)$ on LHC, while the asymmetric property remains inconclusive in both cases.

These results highlight the practical relevance of symmetric robustness, independently of whether the asymmetric property holds. Our new property enables meaningful robustness guarantees in parameter regimes where the asymmetric property cannot be established. This may be due to inconclusive verification or because the asymmetric requirement is too strong and may be violated, even though such violations cannot be confirmed by current verifiers.

\subsection{Parameter Sensitivity}

We analyze the effect of confidence threshold $\tau$ and perturbation bound $\varepsilon$ on verification time and conclusiveness.

\subsubsection{Effect of $\tau$}
Higher $\tau$ generally reduces verification time, as higher confidence constraints eliminate larger parts of the search space. This is clearly visible in the ACAS benchmark: for $\varepsilon = 0.001$, asymmetric verification at $\tau = 0.99$ takes 182 sec, compared to 6742 sec at $\tau = 0.7$, while at 
$\tau = 0.6$ it times out entirely. The same trend holds for the symmetric property, where verification takes 146 sec at $\tau = 0.99$, 6355 sec at $\tau = 0.7$, and 6914 sec at $\tau = 0.6$. This suggests a smooth trade-off: as $\tau$ increases, the region of input space that must be proven robust shrinks, making the problem more tractable.

\subsubsection{Effect of $\varepsilon$}
\looseness=-1
Larger $\varepsilon$ increases the size of secondary generators, leading to greater divergence between the two executions and a higher likelihood of requiring generator splits. On the HAR benchmark, \ourTool{} proves safety for all $\tau \geq 0.6$ at $\varepsilon = 0.001$. For $\varepsilon = 0.01$, the approach times out or exhausts memory for $\tau \leq 0.8$. Notably, at $\tau = 0.9$ with $\varepsilon = 0.01$, the symmetric property is verifiable while the asymmetric variant times out, demonstrating the practical value of the symmetric formulation for larger perturbations.

\subsection{Discussion}
\looseness=-1
Our results demonstrate that our approach implemented in \ourTool{} is both faster and more precise than the CGR baseline, while scaling more effectively to larger DNNs. Across the ACAS, HAR, and LHC benchmarks, our tool establishes confidence-based robustness in settings where the baseline either times out or reports spurious counterexamples due to the coarse softmax approximations.
In particular, for larger DNNs with many output classes, the baseline approach suffers from class-dependent approximation errors and scalability limitations, whereas \ourTool{} maintains efficiency and precision.

In contrast to other properties like local robustness, global robustness verification must consider the DNN's entire input space (rather than e.g.\ only considering minor perturbations around a single fixed datapoint).
Despite this, \ourTool{} is able to scale to a NN with more than 500 input dimensions (HAR), underscoring the scalability advancements achieved by differential verification.

Our experiments also show that our symmetric confidence-based robustness property provides tangible practical benefits on both HAR and LHC benchmarks.
The symmetric formulation enables successful verification in parameter configurations where the asymmetric property remains inconclusive.
By requiring confidence constraints on \emph{both} inputs, the symmetric property allows us to recover meaningful robustness guarantees without sacrificing scalability.
For example, symmetric confidence-based robustness enabled us to establish a global robustness guarantee w.r.t. perturbations of $1/200^{\text{th}}$ of the sensor ranges of HAR's 561-dimensional input space.

We note that refinement is used adaptively. In many cases (e.g., all verified HAR properties), we can verify robustness properties without splitting, underlining the precision of our abstract domain. While on ACAS, verifying the symmetric property at $\varepsilon = 0.001$ requires between 2838 (at $\tau = 0.99$) and 168798 splits (at $\tau = 0.6$).
Overall, our results highlight the combined advantages of our property and verifier in terms of efficiency, precision, and expressiveness.

\paragraph{Limitations.}
\looseness=-1
Confidence-based robustness has an inherent risk of vacuous guarantees when robustness is established for confidence levels unattainable within the input domain.
\ourTool{} has no explicit witness search procedure, but we check propagated zonotope centers as a lightweight check for confidence attainability.
For all benchmarks except ACAS, this procedure found concrete witnesses.
Despite this epistemic uncertainty for ACAS, we include the ACAS results as it is a widely used and sufficiently complex VNN-COMP benchmark.

\section{Related work}
\label{related}

\subsubsection*{Global Robustness Verification.}
The work closest in spirit to ours is by Athavale \emph{et al.}~\cite{athavale2024verifying}.
They introduce the notion of asymmetric confidence-based global robustness cited in \Cref{def:conf_rob_old}.
We compare against their approach as \emph{CGR baseline} which uses self-composition and employs the verifier Marabou~\cite{katz2019marabou} for verification.
They rely on a piecewise linear approximation of softmax-based confidence constraint and verify robustness by exhaustively checking class-wise dominance conditions.
While their method provides strong correctness guarantees, its reliance on complete solving and class-dependent reasoning limits scalability, particularly as the number of output classes grows.

\subsubsection*{Local Robustness Verification.} 
A large body of research in the DNN safety domain tackles local robustness. A network is locally robust w.r.t. to a given input $\boldsymbol{x}$ if all the points in the vicinity of $\boldsymbol{x}$ map to the same output as $\boldsymbol{x}$.

Most techniques for the verification of local robustness properties either rely on static analysis~\cite{pulina2010abstraction, singh2019abstract,Singh18,gehr2018ai2,10097028, wang2021beta}
via over-approximations or SMT- or mixed integer programming-based approaches~\cite{seshia2018formal,huang2017safety,gopinath2018deepsafe, cheng2017maximum, tjeng2017evaluating,dutta2018output, katz2017reluplex,katz2019marabou}  which offer greater precision at the cost of reduced scalability.
Bosman \emph{et al.}~\cite{DBLP:journals/jair/BosmanBHR25} accumulate local robustness results for many such $\boldsymbol{x}$ to characterize a DNN's overall robustness.
Orthogonally, Geng \emph{et al.}~\cite{DBLP:conf/icml/GengLXWGS23} identify input space regions with identical neural activation patterns and verify classification consistency.
Unlike our approach, all mentioned works stop short of providing guarantees over the \emph{entire} input space.

\subsubsection*{Other Relational Properties.}

Recent works by Banerjee et al.~\cite{banerjee2024input} and Suresh et al.~\cite{suresh2024relational} study input-relational verification in the context of universal adversarial perturbations~\cite{DBLP:conf/cvpr/Moosavi-Dezfooli17} (UAPs).
Their formulation considers a fixed set of inputs and asks whether a \emph{single} perturbation, applied uniformly, can cause all of them to be misclassified.
The relational aspect thus arises from sharing a perturbation across discrete inputs, rather than from semantic relationships between the inputs themselves.
In contrast, we reason about \emph{global} 2-safety properties over a DNN’s \emph{entire} input space, achieving this global perspective by incorporating prediction confidence.

At a high level, our work is related to fairness verifiers. Khedr et al.~\cite{khedr2022certifair} and Biswas et al.~\cite{BiswasR23} verify individual and group fairness in DNNs. Both frameworks reason about fairness constraints defined with respect to sensitive attributes or population-level parity notions, and verify whether these constraints hold within structured subsets of the input space. FairQuant \cite{10.1109/ICSE55347.2025.00016} proposes a quantitative certification procedure for 
fairness of DNNs, computing a lower bound on the certified percentage 
of fair inputs. To this end, FairQuant 
partitions the input space and reports the fraction of partitions 
verified to be fair. In contrast to our work, FairQuant does not 
provide a \emph{global} guarantee over the entire input space, but 
rather a quantitative approximation of the fraction of inputs 
satisfying the fairness property.

QEBVerif proposes a verification method for quantization error bounds in fully quantized neural networks. Like our approach, QEBVerif performs 
differential analysis by propagating the difference between two 
network executions layer by layer, rather than independently computing 
output intervals and subtracting. However, QEBVerif addresses a 
structurally different problem: it compares a DNN with its quantized 
counterpart on related inputs, where the QNN input fully determines 
the DNN input via a fixed normalization. This means that they have one universally quantified variable, unlike our setting, where both inputs are universally quantified over the entire input space.

\subsubsection*{Zonotopes.}

Zonotopes are a well-established abstract domain for DNN verification used for safety and local robustness properties, for example in tools such as DeepZ~\cite{Singh18}, NNV~\cite{DBLP:conf/cav/TranYLMNXBJ20,DBLP:conf/cav/LopezCTJ23}, or nnenum~\cite{bak2020improved,DBLP:conf/nfm/Bak21}.
Some works leverage polynomial~\cite{DBLP:conf/nfm/Kochdumper0AB23,DBLP:conf/aaai/LadnerA24} or hybrid~\cite{DBLP:conf/amcc/ZhangX23} zonotopes to increase the precision of the abstract domain.
While these works focus on single-DNN properties, nnenum has also been extended to equivalence verification~\cite{DBLP:conf/ictai/TeuberBKS21} without differential verification.

\looseness=-1
Teuber \emph{et al.}~\cite{teuber2025revisiting} use zonotopes for differential verification~\cite{paulsen_reludiff_2020,paulsen_neurodiff_2020} of equivalence by bounding reachable values of individual DNN executions and their differences.
In contrast, we do not address equivalence (same input, different DNNs), but two-safety properties (different inputs, same DNN), which require significant changes to the verification algorithm (see \Cref{sec:algorithm}).
We demonstrate that zonotope-based differential verification significantly improves performance for confidence-based global robustness properties.

\section{Conclusion}
\label{conclusion}

We present a scalable framework for verifying confidence-based global robustness properties in DNNs. First, we introduce a novel property, \emph{symmetric confidence-based global robustness}, which relaxes the restrictive constraints of asymmetric confidence-based global robustness by permitting adversarial examples when they result in low-confidence predictions. Secondly, we take a first step towards bridging the gap between precision and scalability with \emph{differential halo zonotopes}, a novel technique that propagates perturbed input pairs while precisely bounding differences between DNN executions.

Our experimental evaluation confirms that our approach outperforms the baseline and that \ourTool{} scales to significantly wider and deeper networks than prior methods, verifying significantly wider (500 ReLUs) and deeper DNNs (6 times 50 ReLUs) while supporting larger input spaces (up to 561 dimensions).

While this work focuses on proving safety, extending the framework with counterexample search strategies remains a natural and promising direction to accelerate the search for adversarial attacks.
Another interesting research direction is to devise compositionality techniques to further scale up the  static analysis to larger multidimensional input spaces and more complex DNNs.

\bibliographystyle{ACM-Reference-Format}
\bibliography{referencesclean}

@inproceedings{ao2023two,
  author       = {Shuang Ao and
                  Stefan Rueger and
                  Advaith Siddharthan},
  editor       = {Robin J. Evans and
                  Ilya Shpitser},
  title        = {Two Sides of Miscalibration: Identifying Over and Under-Confidence
                  Prediction for Network Calibration},
  booktitle    = {Uncertainty in Artificial Intelligence},
  eventtitleaddon = {{UAI} 2023},
  series       = {Proceedings of Machine Learning Research},
  volume       = {216},
  pages        = {77--87},
  publisher    = {{PMLR}},
  year         = {2023},
  url          = {https://proceedings.mlr.press/v216/ao23a.html},
  bibsource    = {dblp computer science bibliography, https://dblp.org}
}

@inproceedings{athavale2024verifying,
  author       = {Anagha Athavale and
                  Ezio Bartocci and
                  Maria Christakis and
                  Matteo Maffei and
                  Dejan Nickovic and
                  Georg Weissenbacher},
  editor       = {Arie Gurfinkel and
                  Vijay Ganesh},
  title        = {Verifying Global Two-Safety Properties in Neural Networks with Confidence},
  booktitle    = {36th International Conference on Computer Aided Verification},
  eventtitleaddon = {{CAV} 2024},
  series       = {{LNCS}},
  volume       = {14682},
  pages        = {329--351},
  publisher    = {Springer},
  year         = {2024},
  opturl          = {https://doi.org/10.1007/978-3-031-65630-9\_17},
  doi          = {10.1007/978-3-031-65630-9\_17},
  bibsource    = {dblp computer science bibliography, https://dblp.org}
}

@inproceedings{bak2020improved,
  author       = {Stanley Bak and
                  Hoang{-}Dung Tran and
                  Kerianne Hobbs and
                  Taylor T. Johnson},
  editor       = {Shuvendu K. Lahiri and
                  Chao Wang},
  title        = {Improved Geometric Path Enumeration for Verifying ReLU Neural Networks},
  booktitle    = {32nd International Conference on Computer Aided Verification},
  eventtitleaddon = {{CAV} 2020},
  series       = {{LNCS}},
  volume       = {12224},
  pages        = {66--96},
  publisher    = {Springer},
  year         = {2020},
  opturl          = {https://doi.org/10.1007/978-3-030-53288-8\_4},
  doi          = {10.1007/978-3-030-53288-8\_4},
  bibsource    = {dblp computer science bibliography, https://dblp.org}
}

@inproceedings{DBLP:conf/nfm/Bak21,
  author       = {Stanley Bak},
  editor       = {Aaron Dutle and
                  Mariano M. Moscato and
                  Laura Titolo and
                  C{\'{e}}sar A. Mu{\~{n}}oz and
                  Ivan Perez},
  title        = {nnenum: Verification of ReLU Neural Networks with Optimized Abstraction
                  Refinement},
  booktitle    = {13th {NASA} Formal Methods Symposium},
  eventtitleaddon = {{NFM} 2021},
  series       = {{LNCS}},
  volume       = {12673},
  pages        = {19--36},
  publisher    = {Springer},
  year         = {2021},
  opturl          = {https://doi.org/10.1007/978-3-030-76384-8\_2},
  doi          = {10.1007/978-3-030-76384-8\_2},
  bibsource    = {dblp computer science bibliography, https://dblp.org}
}

@inproceedings{10.5555/3692070.3692181,
  author       = {Debangshu Banerjee and
                  Gagandeep Singh},
  title        = {Relational {DNN} Verification With Cross Executional Bound Refinement},
  booktitle    = {41st International Conference on Machine Learning},
  eventtitleaddon = {{ICML} 2024},
  publisher    = {OpenReview.net},
  year         = {2024},
  url          = {https://openreview.net/forum?id=HOG80Yk4Gw},
  bibsource    = {dblp computer science bibliography, https://dblp.org}
}

@article{banerjee2024input,
  author       = {Debangshu Banerjee and
                  Changming Xu and
                  Gagandeep Singh},
  title        = {Input-Relational Verification of Deep Neural Networks},
  journal      = {Proc. {ACM} Program. Lang.},
  volume       = {8},
  number       = {{PLDI}},
  pages        = {1--27},
  year         = {2024},
  opturl          = {https://doi.org/10.1145/3656377},
  doi          = {10.1145/3656377},
  bibsource    = {dblp computer science bibliography, https://dblp.org}
}

@inproceedings{10097028,
  author       = {Anahita Baninajjar and
                  Kamran Hosseini and
                  Ahmed Rezine and
                  Amir Aminifar},
  title        = {SafeDeep: {A} Scalable Robustness Verification Framework for Deep
                  Neural Networks},
  booktitle    = {{IEEE} International Conference on Acoustics, Speech and Signal Processing},
  eventtitleaddon = {{IEEE} {ICASSP} 2023},
  pages        = {1--5},
  publisher    = {{IEEE}},
  year         = {2023},
  opturl          = {https://doi.org/10.1109/ICASSP49357.2023.10097028},
  doi          = {10.1109/ICASSP49357.2023.10097028},
  bibsource    = {dblp computer science bibliography, https://dblp.org}
}

@article{bezanson2017julia,
  author       = {Jeff Bezanson and
                  Alan Edelman and
                  Stefan Karpinski and
                  Viral B. Shah},
  title        = {Julia: {A} Fresh Approach to Numerical Computing},
  journal      = {{SIAM} Rev.},
  volume       = {59},
  number       = {1},
  pages        = {65--98},
  year         = {2017},
  opturl          = {https://doi.org/10.1137/141000671},
  doi          = {10.1137/141000671},
  bibsource    = {dblp computer science bibliography, https://dblp.org}
}

@inproceedings{BiswasR23,
  author       = {Sumon Biswas and
                  Hridesh Rajan},
  title        = {Fairify: Fairness Verification of Neural Networks},
  booktitle    = {45th {IEEE/ACM} International Conference on Software Engineering,
                  {ICSE}},
  eventtitleaddon = {{IEEE/ACM} {ICSE} 2023},
  pages        = {1546--1558},
  publisher    = {{IEEE}},
  year         = {2023},
  url          = {https://doi.org/10.1109/ICSE48619.2023.00134},
  doi          = {10.1109/ICSE48619.2023.00134},
  bibsource    = {dblp computer science bibliography, https://dblp.org}
}

@article{DBLP:journals/jair/BosmanBHR25,
  author       = {Annelot W. Bosman and
                  Aaron Berger and
                  Holger H. Hoos and
                  Jan N. van Rijn},
  title        = {Robustness Distributions in Neural Network Verification},
  journal      = {J. Artif. Intell. Res.},
  volume       = {83},
  year         = {2025},
  url          = {https://doi.org/10.1613/jair.1.18403},
  doi          = {10.1613/JAIR.1.18403},
  bibsource    = {dblp computer science bibliography, https://dblp.org}
}

@inproceedings{cheng2017maximum,
  author       = {Chih{-}Hong Cheng and
                  Georg N{\"{u}}hrenberg and
                  Harald Ruess},
  editor       = {Deepak D'Souza and
                  K. Narayan Kumar},
  title        = {Maximum Resilience of Artificial Neural Networks},
  booktitle    = {15th International Symposium on Automated Technology for Verification and Analysis},
  eventtitleaddon = {{ATVA} 2017},
  series       = {{LNCS}},
  volume       = {10482},
  pages        = {251--268},
  publisher    = {Springer},
  year         = {2017},
  opturl          = {https://doi.org/10.1007/978-3-319-68167-2\_18},
  doi          = {10.1007/978-3-319-68167-2\_18},
  bibsource    = {dblp computer science bibliography, https://dblp.org}
}

@article{duarte2018fast,
  title={Fast inference of deep neural networks in FPGAs for particle physics},
  author={Duarte, Javier and Han, Song and Harris, Philip and Jindariani, Sergo and Kreinar, Edward and Kreis, Benjamin and Ngadiuba, Jennifer and Pierini, Maurizio and Rivera, Ryan and Tran, Nhan and others},
  journal={Journal of instrumentation},
  doi={10.1088/1748-0221/13/07/P07027},
  eprint = {1804.06913},
  volume={13},
  number={07},
  pages={P07027},
  year={2018},
  publisher={IOP Publishing}
}

@inproceedings{dutta2018output,
  author       = {Souradeep Dutta and
                  Susmit Jha and
                  Sriram Sankaranarayanan and
                  Ashish Tiwari},
  editor       = {Aaron Dutle and
                  C{\'{e}}sar A. Mu{\~{n}}oz and
                  Anthony Narkawicz},
  title        = {Output Range Analysis for Deep Feedforward Neural Networks},
  booktitle    = {10th {NASA} Formal Methods Symposium},
  eventtitleaddon = {{NFM} 2028},
  series       = {{LNCS}},
  volume       = {10811},
  pages        = {121--138},
  publisher    = {Springer},
  year         = {2018},
  opturl          = {https://doi.org/10.1007/978-3-319-77935-5\_9},
  doi          = {10.1007/978-3-319-77935-5\_9},
  bibsource    = {dblp computer science bibliography, https://dblp.org}
}

@inproceedings{gehr2018ai2,
  author       = {Timon Gehr and
                  Matthew Mirman and
                  Dana Drachsler{-}Cohen and
                  Petar Tsankov and
                  Swarat Chaudhuri and
                  Martin T. Vechev},
  title        = {{AI2:} Safety and Robustness Certification of Neural Networks with
                  Abstract Interpretation},
  booktitle    = {{IEEE} Symposium on Security and Privacy},
  eventtitleaddon = {{IEEE} {SP} 2018},
  pages        = {3--18},
  publisher    = {{IEEE} Computer Society},
  year         = {2018},
  opturl          = {https://doi.org/10.1109/SP.2018.00058},
  doi          = {10.1109/SP.2018.00058},
  bibsource    = {dblp computer science bibliography, https://dblp.org}
}

@inproceedings{DBLP:conf/cav/GhorbalGP09,
  author       = {Khalil Ghorbal and
                  Eric Goubault and
                  Sylvie Putot},
  editor       = {Ahmed Bouajjani and
                  Oded Maler},
  title        = {The Zonotope Abstract Domain Taylor1+},
  booktitle    = {21st International Conference on Computer Aided Verification},
  eventtitleaddon = {{CAV} 2009},
  series       = {{LNCS}},
  volume       = {5643},
  pages        = {627--633},
  publisher    = {Springer},
  year         = {2009},
  opturl          = {https://doi.org/10.1007/978-3-642-02658-4\_47},
  doi          = {10.1007/978-3-642-02658-4\_47},
  bibsource    = {dblp computer science bibliography, https://dblp.org}
}

@book{goodfellow2016deep,
  author    = {Ian J. Goodfellow and
               Yoshua Bengio and
               Aaron C. Courville},
  title     = {Deep Learning},
  series    = {Adaptive computation and machine learning},
  publisher = {{MIT} Press},
  year      = {2016},
  url       = {http://www.deeplearningbook.org/},
  isbn      = {978-0-262-03561-3},
}

@inproceedings{gopinath2018deepsafe,
  author       = {Divya Gopinath and
                  Guy Katz and
                  Corina S. Pasareanu and
                  Clark W. Barrett},
  editor       = {Shuvendu K. Lahiri and
                  Chao Wang},
  title        = {DeepSafe: {A} Data-Driven Approach for Assessing Robustness of Neural
                  Networks},
  booktitle    = {16th International Symposium on Automated Technology for Verification and Analysis},
  eventtitleaddon = {{ATVA} 2018},
  series       = {{LNCS}},
  volume       = {11138},
  pages        = {3--19},
  publisher    = {Springer},
  year         = {2018},
  opturl          = {https://doi.org/10.1007/978-3-030-01090-4\_1},
  doi          = {10.1007/978-3-030-01090-4\_1},
  bibsource    = {dblp computer science bibliography, https://dblp.org}
}

@inproceedings{guo2017calibration,
  author       = {Chuan Guo and
                  Geoff Pleiss and
                  Yu Sun and
                  Kilian Q. Weinberger},
  editor       = {Doina Precup and
                  Yee Whye Teh},
  title        = {On Calibration of Modern Neural Networks},
  booktitle    = {34th International Conference on Machine Learning},
  eventtitleaddon = {{ICML} 2017},
  series       = {Proceedings of Machine Learning Research},
  volume       = {70},
  pages        = {1321--1330},
  publisher    = {{PMLR}},
  year         = {2017},
  url          = {http://proceedings.mlr.press/v70/guo17a.html},
  bibsource    = {dblp computer science bibliography, https://dblp.org}
}

@misc{gurobi,
  author = {{Gurobi Optimization, LLC}},
  title = {{Gurobi Optimizer Reference Manual}},
  year = 2023,
  url = "https://www.gurobi.com"
}

@inproceedings{huang2017safety,
  author       = {Xiaowei Huang and
                  Marta Kwiatkowska and
                  Sen Wang and
                  Min Wu},
  editor       = {Rupak Majumdar and
                  Viktor Kuncak},
  title        = {Safety Verification of Deep Neural Networks},
  booktitle    = {29th International Conference on Computer Aided Verification},
  eventtitleaddon = {{CAV} 2017},
  series       = {{LNCS}},
  volume       = {10426},
  pages        = {3--29},
  publisher    = {Springer},
  year         = {2017},
  opturl          = {https://doi.org/10.1007/978-3-319-63387-9\_1},
  doi          = {10.1007/978-3-319-63387-9\_1},
  bibsource    = {dblp computer science bibliography, https://dblp.org}
}

@inproceedings{julian2016policy,
  author={Julian, Kyle D. and Lopez, Jessica and Brush, Jeffrey S. and Owen, Michael P. and Kochenderfer, Mykel J.},
  booktitle={35th IEEE/AIAA Digital Avionics Systems Conference},
  eventtitleaddon = {{DASC} 2016},
  title={Policy compression for aircraft collision avoidance systems},  
  year={2016},
  pages={1-10},
  doi={10.1109/DASC.2016.7778091}
}

@inproceedings{katz2017reluplex,
  author       = {Guy Katz and
                  Clark W. Barrett and
                  David L. Dill and
                  Kyle Julian and
                  Mykel J. Kochenderfer},
  editor       = {Rupak Majumdar and
                  Viktor Kuncak},
  title        = {{Reluplex}: An Efficient {SMT} Solver for Verifying Deep Neural Networks},
  booktitle    = {29th International Conference on Computer Aided Verification},
  eventtitleaddon = {{CAV} 2017},
  series       = {{LNCS}},
  volume       = {10426},
  pages        = {97--117},
  publisher    = {Springer},
  year         = {2017},
  opturl          = {https://doi.org/10.1007/978-3-319-63387-9\_5},
  doi          = {10.1007/978-3-319-63387-9\_5},
  bibsource    = {dblp computer science bibliography, https://dblp.org}
}

@inproceedings{katz2019marabou,
  author       = {Guy Katz and
                  Derek A. Huang and
                  Duligur Ibeling and
                  Kyle Julian and
                  Christopher Lazarus and
                  Rachel Lim and
                  Parth Shah and
                  Shantanu Thakoor and
                  Haoze Wu and
                  Aleksandar Zeljic and
                  David L. Dill and
                  Mykel J. Kochenderfer and
                  Clark W. Barrett},
  editor       = {Isil Dillig and
                  Serdar Tasiran},
  title        = {The Marabou Framework for Verification and Analysis of Deep Neural
                  Networks},
  booktitle    = {31st International Conference on Computer Aided Verification},
  eventtitleaddon = {{CAV} 2019},
  series       = {{LNCS}},
  volume       = {11561},
  pages        = {443--452},
  publisher    = {Springer},
  year         = {2019},
  opturl          = {https://doi.org/10.1007/978-3-030-25540-4\_26},
  doi          = {10.1007/978-3-030-25540-4\_26},
  bibsource    = {dblp computer science bibliography, https://dblp.org}
}

@inproceedings{kern2025certified,
  author       = {Philipp Kern and
                  Edoardo Manino and
                  Carsten Sinz},
  editor       = {Mirco Giacobbe and
                  Anna Lukina},
  title        = {Certified Error Analysis of Homomorphically Encrypted Neural Networks},
  booktitle    = {2nd International Symposium on {AI} Verification},
  eventtitleaddon = {{SAIV} 2025},
  series       = {{LNCS}},
  volume       = {15947},
  pages        = {156--179},
  publisher    = {Springer},
  year         = {2025},
  opturl          = {https://doi.org/10.1007/978-3-031-99991-8\_8},
  doi          = {10.1007/978-3-031-99991-8\_8},
  bibsource    = {dblp computer science bibliography, https://dblp.org}
}

@inproceedings{khedr2022certifair,
  author       = {Haitham Khedr and
                  Yasser Shoukry},
  editor       = {Brian Williams and
                  Yiling Chen and
                  Jennifer Neville},
  title        = {{CertiFair}: {A} Framework for Certified Global Fairness of Neural Networks},
  booktitle    = {37th {AAAI} Conference on Artificial Intelligence},
  eventtitleaddon = {{AAAI} 2023},
  pages        = {8237--8245},
  publisher    = {{AAAI} Press},
  year         = {2023},
  opturl          = {https://doi.org/10.1609/aaai.v37i7.25994},
  doi          = {10.1609/AAAI.V37I7.25994},
  bibsource    = {dblp computer science bibliography, https://dblp.org}
}

@inproceedings{kull2019beyond,
  author       = {Meelis Kull and
                  Miquel Perell{\'{o}}{-}Nieto and
                  Markus K{\"{a}}ngsepp and
                  Telmo de Menezes e Silva Filho and
                  Hao Song and
                  Peter A. Flach},
  editor       = {Hanna M. Wallach and
                  Hugo Larochelle and
                  Alina Beygelzimer and
                  Florence d'Alch{\'{e}}{-}Buc and
                  Emily B. Fox and
                  Roman Garnett},
  title        = {Beyond temperature scaling: Obtaining well-calibrated multi-class
                  probabilities with Dirichlet calibration},
  booktitle    = {Annual Conference on Neural Information Processing Systems},
  eventtitleaddon = {{NeurIPS} 2019},
  pages        = {12295--12305},
  year         = {2019},
  opturl          = {https://proceedings.neurips.cc/paper/2019/hash/8ca01ea920679a0fe3728441494041b9-Abstract.html},
  bibsource    = {dblp computer science bibliography, https://dblp.org}
}

@inproceedings{li2019analyzing,
  author       = {Jianlin Li and
                  Jiangchao Liu and
                  Pengfei Yang and
                  Liqian Chen and
                  Xiaowei Huang and
                  Lijun Zhang},
  editor       = {Bor{-}Yuh Evan Chang},
  title        = {Analyzing Deep Neural Networks with Symbolic Propagation: Towards
                  Higher Precision and Faster Verification},
  booktitle    = {26th International Symposium on Static Analysis},
  eventtitleaddon = {{SAS} 2019},
  series       = {{LNCS}},
  volume       = {11822},
  pages        = {296--319},
  publisher    = {Springer},
  year         = {2019},
  opturl          = {https://doi.org/10.1007/978-3-030-32304-2\_15},
  doi          = {10.1007/978-3-030-32304-2\_15},
  bibsource    = {dblp computer science bibliography, https://dblp.org}
}

@inproceedings{DBLP:conf/cav/LopezCTJ23,
  author       = {Diego Manzanas Lopez and
                  Sung Woo Choi and
                  Hoang{-}Dung Tran and
                  Taylor T. Johnson},
  editor       = {Constantin Enea and
                  Akash Lal},
  title        = {{NNV} 2.0: The Neural Network Verification Tool},
  booktitle    = {35th International Conference on Computer Aided Verification},
  eventtitleaddon = {{CAV} 2023},
  series       = {{LNCS}},
  volume       = {13965},
  pages        = {397--412},
  publisher    = {Springer},
  year         = {2023},
  opturl          = {https://doi.org/10.1007/978-3-031-37703-7\_19},
  doi          = {10.1007/978-3-031-37703-7\_19},
  bibsource    = {dblp computer science bibliography, https://dblp.org}
}

@inproceedings{DBLP:conf/cvpr/Moosavi-Dezfooli17,
  author       = {Seyed{-}Mohsen Moosavi{-}Dezfooli and
                  Alhussein Fawzi and
                  Omar Fawzi and
                  Pascal Frossard},
  title        = {Universal Adversarial Perturbations},
  booktitle    = {{IEEE} Conference on Computer Vision and Pattern Recognition},
  eventtitleaddon = {{IEEE} {CVPR} 2017},
  pages        = {86--94},
  publisher    = {{IEEE} Computer Society},
  year         = {2017},
  opturl          = {https://doi.org/10.1109/CVPR.2017.17},
  doi          = {10.1109/CVPR.2017.17},
  bibsource    = {dblp computer science bibliography, https://dblp.org}
}

@inproceedings{paulsen_neurodiff_2020,
  author    = {Brandon Paulsen and
               Jingbo Wang and
               Jiawei Wang and
               Chao Wang},
  title     = {{NeuroDiff:} Scalable Differential Verification of Neural Networks
               using Fine-Grained Approximation},
  booktitle = {35th {IEEE/ACM} International Conference on Automated Software Engineering},
  eventtitleaddon = {{IEEE/ACM} {ASE} 2020},
  pages     = {784--796},
  publisher = {{IEEE}},
  year      = {2020},
  opturl       = {https://doi.org/10.1145/3324884.3416560},
  doi       = {10.1145/3324884.3416560},
  bibsource = {dblp computer science bibliography, https://dblp.org}
}

@inproceedings{paulsen_reludiff_2020,
  author       = {Brandon Paulsen and
                  Jingbo Wang and
                  Chao Wang},
  editor       = {Gregg Rothermel and
                  Doo{-}Hwan Bae},
  title        = {ReluDiff: differential verification of deep neural networks},
  booktitle    = {42nd International Conference on Software Engineering},
  eventtitleaddon = {{IEEE/ACM} {ICSE} 2020},
  pages        = {714--726},
  publisher    = {{ACM}},
  year         = {2020},
  opturl          = {https://doi.org/10.1145/3377811.3380337},
  doi          = {10.1145/3377811.3380337},
  bibsource    = {dblp computer science bibliography, https://dblp.org}
}

@book{DBLP:books/sp/Paulson94,
  author       = {Lawrence C. Paulson},
  title        = {Isabelle - {A} Generic Theorem Prover (with a contribution by T. Nipkow)},
  series       = {{LNCS}},
  volume       = {828},
  publisher    = {Springer},
  year         = {1994},
  opturl          = {https://doi.org/10.1007/BFb0030541},
  doi          = {10.1007/BFB0030541},
  isbn         = {3-540-58244-4},
  bibsource    = {dblp computer science bibliography, https://dblp.org}
}

@inproceedings{pulina2010abstraction,
  author       = {Luca Pulina and
                  Armando Tacchella},
  editor       = {Tayssir Touili and
                  Byron Cook and
                  Paul B. Jackson},
  title        = {An Abstraction-Refinement Approach to Verification of Artificial Neural
                  Networks},
  booktitle    = {22nd International Conference on Computer Aided Verification},
  eventtitleaddon = {{CAV} 2010},
  series       = {{LNCS}},
  volume       = {6174},
  pages        = {243--257},
  publisher    = {Springer},
  year         = {2010},
  opturl          = {https://doi.org/10.1007/978-3-642-14295-6\_24},
  doi          = {10.1007/978-3-642-14295-6\_24},
  bibsource    = {dblp computer science bibliography, https://dblp.org}
}

@article{pulina2012challenging,
  author       = {Luca Pulina and
                  Armando Tacchella},
  title        = {Challenging {SMT} solvers to verify neural networks},
  journal      = {{AI} Commun.},
  volume       = {25},
  number       = {2},
  pages        = {117--135},
  year         = {2012},
  opturl          = {https://doi.org/10.3233/AIC-2012-0525},
  doi          = {10.3233/AIC-2012-0525},
  bibsource    = {dblp computer science bibliography, https://dblp.org}
}

@inproceedings{seshia2018formal,
  author       = {Sanjit A. Seshia and
                  Ankush Desai and
                  Tommaso Dreossi and
                  Daniel J. Fremont and
                  Shromona Ghosh and
                  Edward Kim and
                  Sumukh Shivakumar and
                  Marcell Vazquez{-}Chanlatte and
                  Xiangyu Yue},
  editor       = {Shuvendu K. Lahiri and
                  Chao Wang},
  title        = {Formal Specification for Deep Neural Networks},
  booktitle    = {16th International Symposium on Automated Technology for Verification and Analysis},
  eventtitleaddon = {{ATVA} 2018},
  series       = {{LNCS}},
  volume       = {11138},
  pages        = {20--34},
  publisher    = {Springer},
  year         = {2018},
  opturl          = {https://doi.org/10.1007/978-3-030-01090-4\_2},
  doi          = {10.1007/978-3-030-01090-4\_2},
  bibsource    = {dblp computer science bibliography, https://dblp.org}
}

@article{singh2019abstract,
  author       = {Gagandeep Singh and
                  Timon Gehr and
                  Markus P{\"{u}}schel and
                  Martin T. Vechev},
  title        = {An abstract domain for certifying neural networks},
  journal      = {Proc. {ACM} Program. Lang.},
  volume       = {3},
  number       = {{POPL}},
  pages        = {41:1--41:30},
  year         = {2019},
  opturl          = {https://doi.org/10.1145/3290354},
  doi          = {10.1145/3290354},
  bibsource    = {dblp computer science bibliography, https://dblp.org}
}

@inproceedings{Singh18,
  author       = {Gagandeep Singh and
                  Timon Gehr and
                  Matthew Mirman and
                  Markus P{\"{u}}schel and
                  Martin T. Vechev},
  editor       = {Samy Bengio and
                  Hanna M. Wallach and
                  Hugo Larochelle and
                  Kristen Grauman and
                  Nicol{\`{o}} Cesa{-}Bianchi and
                  Roman Garnett},
  title        = {Fast and Effective Robustness Certification},
  booktitle    = {Annual Conference on Neural Information Processing Systems},
  eventtitleaddon = {{NeurIPS} 2018},
  pages        = {10825--10836},
  year         = {2018},
  opturl          = {https://proceedings.neurips.cc/paper/2018/hash/f2f446980d8e971ef3da97af089481c3-Abstract.html},
  bibsource    = {dblp computer science bibliography, https://dblp.org}
}

@inproceedings{suresh2024relational,
  author       = {Tarun Suresh and
                  Debangshu Banerjee and
                  Gagandeep Singh},
  editor       = {Amir Globersons and
                  Lester Mackey and
                  Danielle Belgrave and
                  Angela Fan and
                  Ulrich Paquet and
                  Jakub M. Tomczak and
                  Cheng Zhang},
  title        = {Relational Verification Leaps Forward with RABBit},
  booktitle    = {Annual Conference on Neural Information Processing Systems},
  eventtitleaddon = {{NeurIPS} 2024},
  year         = {2024},
  opturl          = {http://papers.nips.cc/paper\_files/paper/2024/hash/df15617fd6d6f78fcd485401d0598761-Abstract-Conference.html},
  bibsource    = {dblp computer science bibliography, https://dblp.org}
}

@inproceedings{DBLP:conf/ictai/TeuberBKS21,
  author       = {Samuel Teuber and
                  Marko Kleine B{\"{u}}ning and
                  Philipp Kern and
                  Carsten Sinz},
  title        = {Geometric Path Enumeration for Equivalence Verification of Neural
                  Networks},
  booktitle    = {33rd {IEEE} International Conference on Tools with Artificial Intelligence},
  eventtitleaddon = {{IEEE} {ICTAI} 2021},
  pages        = {200--208},
  publisher    = {{IEEE}},
  year         = {2021},
  opturl          = {https://doi.org/10.1109/ICTAI52525.2021.00035},
  doi          = {10.1109/ICTAI52525.2021.00035},
  bibsource    = {dblp computer science bibliography, https://dblp.org}
}

@inproceedings{teuber2025revisiting,
  author       = {Samuel Teuber and
                  Philipp Kern and
                  Marvin Janzen and
                  Bernhard Beckert},
  editor       = {Arie Gurfinkel and
                  Marijn Heule},
  title        = {Revisiting Differential Verification: Equivalence Verification with
                  Confidence},
  booktitle    = {31st International Conference on Tools and Algorithms for the Construction and Analysis of Systems},
  eventtitleaddon = {{TACAS} 2025},
  series       = {{LNCS}},
  volume       = {15697},
  pages        = {257--278},
  publisher    = {Springer},
  year         = {2025},
  opturl          = {https://doi.org/10.1007/978-3-031-90653-4\_13},
  doi          = {10.1007/978-3-031-90653-4\_13},
  bibsource    = {dblp computer science bibliography, https://dblp.org}
}

@inproceedings{tjeng2017evaluating,
  author       = {Vincent Tjeng and
                  Kai Yuanqing Xiao and
                  Russ Tedrake},
  title        = {Evaluating Robustness of Neural Networks with Mixed Integer Programming},
  booktitle    = {7th International Conference on Learning Representations},
  eventtitleaddon = {{ICLR} 2019},
  publisher    = {OpenReview.net},
  year         = {2019},
  url          = {https://openreview.net/forum?id=HyGIdiRqtm},
  bibsource    = {dblp computer science bibliography, https://dblp.org}
}

@inproceedings{DBLP:conf/cav/TranYLMNXBJ20,
  author       = {Hoang{-}Dung Tran and
                  Xiaodong Yang and
                  Diego Manzanas Lopez and
                  Patrick Musau and
                  Luan Viet Nguyen and
                  Weiming Xiang and
                  Stanley Bak and
                  Taylor T. Johnson},
  editor       = {Shuvendu K. Lahiri and
                  Chao Wang},
  title        = {{NNV:} The Neural Network Verification Tool for Deep Neural Networks
                  and Learning-Enabled Cyber-Physical Systems},
  booktitle    = {32nd International Conference on Computer Aided Verification},
  eventtitleaddon = {{CAV} 2020},
  series       = {{LNCS}},
  volume       = {12224},
  pages        = {3--17},
  publisher    = {Springer},
  year         = {2020},
  opturl          = {https://doi.org/10.1007/978-3-030-53288-8\_1},
  doi          = {10.1007/978-3-030-53288-8\_1},
  bibsource    = {dblp computer science bibliography, https://dblp.org}
}

@inproceedings{wang2021beta,
  author       = {Shiqi Wang and
                  Huan Zhang and
                  Kaidi Xu and
                  Xue Lin and
                  Suman Jana and
                  Cho{-}Jui Hsieh and
                  J. Zico Kolter},
  editor       = {Marc'Aurelio Ranzato and
                  Alina Beygelzimer and
                  Yann N. Dauphin and
                  Percy Liang and
                  Jennifer Wortman Vaughan},
  title        = {Beta-CROWN: Efficient Bound Propagation with Per-neuron Split Constraints
                  for Neural Network Robustness Verification},
  booktitle    = {Annual Conference on Neural Information Processing Systems},
  eventtitleaddon = {{NeurIPS} 2021},
  pages        = {29909--29921},
  year         = {2021},
  opturl          = {https://proceedings.neurips.cc/paper/2021/hash/fac7fead96dafceaf80c1daffeae82a4-Abstract.html},
  bibsource    = {dblp computer science bibliography, https://dblp.org}
}

@misc{wightman1998lsac,
  title={{LSAC} National Longitudinal Bar Passage Study. LSAC Research Report Series.},
  author={Wightman, Linda F},
  year={1998},
  publisher={ERIC}
}

@inproceedings{DBLP:conf/icml/GengLXWGS23,
  author       = {Chuqin Geng and
                  Nham Le and
                  Xiaojie Xu and
                  Zhaoyue Wang and
                  Arie Gurfinkel and
                  Xujie Si},
  editor       = {Andreas Krause and
                  Emma Brunskill and
                  Kyunghyun Cho and
                  Barbara Engelhardt and
                  Sivan Sabato and
                  Jonathan Scarlett},
  title        = {Towards Reliable Neural Specifications},
  booktitle    = {International Conference on Machine Learning},
  eventtitleaddon = {{ICML} 2023},
  series       = {Proceedings of Machine Learning Research},
  volume       = {202},
  pages        = {11196--11212},
  publisher    = {{PMLR}},
  year         = {2023},
  url          = {https://proceedings.mlr.press/v202/geng23a.html},
  bibsource    = {dblp computer science bibliography, https://dblp.org}
}

@inproceedings{DBLP:conf/nfm/Kochdumper0AB23,
  author       = {Niklas Kochdumper and
                  Christian Schilling and
                  Matthias Althoff and
                  Stanley Bak},
  editor       = {Kristin Yvonne Rozier and
                  Swarat Chaudhuri},
  title        = {Open- and Closed-Loop Neural Network Verification Using Polynomial
                  Zonotopes},
  booktitle    = {15th {NASA} Formal Methods Symposium},
  eventtitleaddon = {{NFM} 2023},
  series       = {Lecture Notes in Computer Science},
  volume       = {13903},
  pages        = {16--36},
  publisher    = {Springer},
  year         = {2023},
  opturl          = {https://doi.org/10.1007/978-3-031-33170-1\_2},
  doi          = {10.1007/978-3-031-33170-1\_2},
  bibsource    = {dblp computer science bibliography, https://dblp.org}
}

@inproceedings{DBLP:conf/aaai/LadnerA24,
  author       = {Tobias Ladner and
                  Matthias Althoff},
  editor       = {Michael J. Wooldridge and
                  Jennifer G. Dy and
                  Sriraam Natarajan},
  title        = {Exponent Relaxation of Polynomial Zonotopes and Its Applications in
                  Formal Neural Network Verification},
  booktitle    = {38th {AAAI} Conference on Artificial Intelligence},
  eventtitleaddon = {{AAAI} 2024},
  pages        = {21304--21311},
  publisher    = {{AAAI} Press},
  year         = {2024},
  opturl          = {https://doi.org/10.1609/aaai.v38i19.30125},
  doi          = {10.1609/AAAI.V38I19.30125},
  bibsource    = {dblp computer science bibliography, https://dblp.org}
}

@inproceedings{DBLP:conf/amcc/ZhangX23,
  author       = {Yuhao Zhang and
                  Xiangru Xu},
  title        = {Reachability Analysis and Safety Verification of Neural Feedback Systems
                  via Hybrid Zonotopes},
  booktitle    = {American Control Conference},
  eventtitleaddon = {{ACC} 2023},
  pages        = {1915--1921},
  publisher    = {{IEEE}},
  year         = {2023},
  opturl          = {https://doi.org/10.23919/ACC55779.2023.10156417},
  doi          = {10.23919/ACC55779.2023.10156417},
  bibsource    = {dblp computer science bibliography, https://dblp.org}
}

@misc{kaulen20256th,
  title={The 6th International Verification of Neural Networks Competition (VNN-COMP 2025): Summary and Results},
  author={Kaulen, Konstantin and Ladner, Tobias and Bak, Stanley and Brix, Christopher and Duong, Hai and Flinkow, Thomas and Johnson, Taylor T and Koller, Lukas and Manino, Edoardo and Nguyen, ThanhVu H and others},
  year={2025},
  eprint={2512.19007},
  archivePrefix={arXiv},
  primaryClass={cs.LG},
  url={https://arxiv.org/abs/2512.19007}, 
}

@inproceedings{DBLP:conf/esann/AnguitaGOPR13,
  author       = {Davide Anguita and
                  Alessandro Ghio and
                  Luca Oneto and
                  Xavier Parra and
                  Jorge Luis Reyes{-}Ortiz},
  title        = {A Public Domain Dataset for Human Activity Recognition using Smartphones},
  booktitle    = {21st European Symposium on Artificial Neural Networks},
  eventtitleaddon = {{ESANN} 2013},
  year         = {2013},
  url          = {https://www.esann.org/sites/default/files/proceedings/legacy/es2013-84.pdf},
  bibsource    = {dblp computer science bibliography, https://dblp.org}
}

@inproceedings{DBLP:conf/nips/GeifmanE17,
  author       = {Yonatan Geifman and
                  Ran El{-}Yaniv},
  editor       = {Isabelle Guyon and
                  Ulrike von Luxburg and
                  Samy Bengio and
                  Hanna M. Wallach and
                  Rob Fergus and
                  S. V. N. Vishwanathan and
                  Roman Garnett},
  title        = {Selective Classification for Deep Neural Networks},
  booktitle    = {Advances in Neural Information Processing Systems},
  eventtitleaddon = {{NeurIPS} 2017},
  pages        = {4878--4887},
  year         = {2017},
  opturl          = {https://proceedings.neurips.cc/paper/2017/hash/4a8423d5e91fda00bb7e46540e2b0cf1-Abstract.html},
  bibsource    = {dblp computer science bibliography, https://dblp.org}
}

@article{clarkson2010hyperproperties,
  title={Hyperproperties},
  author={Clarkson, Michael R and Schneider, Fred B},
  journal={Journal of Computer Security},
  volume={18},
  number={6},
  pages={1157--1210},
  year={2010},
  publisher={SAGE Publications Sage UK: London, England}
}

@misc{bibbo2023human,
  title={Human activity recognition (HAR) in healthcare},
  author={Bibbo, Luigi and Vellasco, Marley MBR},
  journal={Applied Sciences},
  volume={13},
  number={24},
  pages={13009},
  year={2023},
  publisher={MDPI}
}

@article{sankaran2025future,
  title={The future of smart healthcare: how ai and har are reshaping hospital workflows},
  author={Sankaran, Ashwin},
  journal={Well Testing J},
  volume={34},
  number={1},
  pages={68--92},
  year={2025}
}

@inproceedings{10.1109/ICSE55347.2025.00016,
author = {Kim, Brian Hyeongseok and Wang, Jingbo and Wang, Chao},
title = {FairQuant: Certifying and Quantifying Fairness of Deep Neural Networks},
year = {2025},
isbn = {9798331505691},
publisher = {IEEE Press},
url = {https://doi.org/10.1109/ICSE55347.2025.00016},
doi = {10.1109/ICSE55347.2025.00016},
booktitle = {Proceedings of the IEEE/ACM 47th International Conference on Software Engineering},
pages = {527–539},
numpages = {13},
location = {Ottawa, Ontario, Canada},
series = {ICSE '25}
}

@article{de2017mobile,
  title={Mobile activity recognition and fall detection system for elderly people using Ameva algorithm},
  author={de la Concepcion, Miguel Angel Alvarez and Morillo, Luis Miguel Soria and Garc{\'\i}a, Juan Antonio {\'A}lvarez and Gonzalez-Abril, Luis},
  journal={Pervasive and Mobile Computing},
  volume={34},
  pages={3--13},
  year={2017},
  publisher={Elsevier}
}

@article{serpush2022wearable,
  title={Wearable sensor-based human activity recognition in the smart healthcare system},
  author={Serpush, Fatemeh and Menhaj, Mohammad Bagher and Masoumi, Behrooz and Karasfi, Babak},
  journal={Computational intelligence and neuroscience},
  volume={2022},
  number={1},
  pages={1391906},
  year={2022},
  publisher={Wiley Online Library}
}

@inproceedings{liu2021overview,
  title={An overview of human activity recognition using wearable sensors: Healthcare and artificial intelligence},
  author={Liu, Rex and Ramli, Albara Ah and Zhang, Huanle and Henricson, Erik and Liu, Xin},
  booktitle={International Conference on Internet of Things},
  pages={1--14},
  year={2021},
  organization={Springer}
}

\appendix

\section{Proofs}
\label{apx:proofs}
\printProofs

\end{document}